\documentclass[12pt]{article}
\usepackage[utf8]{inputenc}
\UseLegacyTextSymbols   % dagger/section footnote marks from the Type 1 math fonts, not from TS1 bitmaps
\usepackage[english]{babel}
\usepackage{amsmath}
\usepackage{amssymb}
\usepackage{amsthm}
\usepackage[margin=1.25in]{geometry}   
\usepackage{enumerate}
\usepackage[shortlabels]{enumitem}
\usepackage{setspace}
\usepackage{graphicx}
\usepackage{float}
\usepackage{subcaption}
\usepackage{tikz}
\usetikzlibrary{arrows.meta, patterns}
\usepackage{multirow}
\usepackage{afterpage}
\usepackage{booktabs}
\usepackage{array}
\usepackage{placeins}
\usepackage{needspace}
\usepackage{etoolbox}
\usepackage{natbib}
\usepackage{xr-hyper}
\usepackage{hyperref}
\newif\iflinkboxes
\linkboxesfalse  % \linkboxestrue: boxed links (for reading); \linkboxesfalse: plain links (for submission)
\iflinkboxes\else\hypersetup{hidelinks}\fi

\usepackage{xcolor}
\newif\ifreview
\reviewfalse
\ifreview
  \newcommand{\new}[1]{{\color{blue}#1}}
  \newcommand{\ser}[1]{{\color{blue}#1}}
  \newenvironment{sergei}{\color{blue}}{}
  \newcommand{\contrib}[1]{{\color{green!45!black}#1}}
  \newenvironment{newresult}{\color{green!45!black}}{}
\else
  \newcommand{\new}[1]{#1}
  \newcommand{\ser}[1]{#1}
  \newenvironment{sergei}{}{}
  \newcommand{\contrib}[1]{#1}
  \newenvironment{newresult}{}{}
\fi

\newcommand{\recommendation}[1]{}
\newcommand{\MEASURE}[1]{}

\usepackage{rotating}

\newcommand{\sug}[1]{#1}
\newcommand{\appref}[1]{Appendix~\ref{#1}}

\newtheorem{theorem}{Theorem}
\newtheorem{proposition}{Proposition}[section]
\newtheorem{lemma}{Lemma}[section]
\newtheorem{corollary}[proposition]{Corollary}

\theoremstyle{remark}
\newtheorem{remark}[proposition]{Remark}
\theoremstyle{definition}
\newtheorem{definition}[proposition]{Definition}
\newtheorem{assumption}{Assumption}
\newtheorem{assumptionH}{Assumption}

\newtheorem{assumptionM}{Assumption}

\newcommand{\cost}{C}

\newcommand{\profit}{\Pi}
\newcommand{\gif}[3]{\Phi(#2,\,#1;\,#3)}
\newcommand{\vs}{f}
\newcommand{\qlim}{q_0}
\newcommand{\qrel}{q_1}
\newcommand{\qopt}{q^*}
\newcommand{\qiso}{q^\nu}

\newcommand{\dens}{p}
\newcommand{\cdf}{P}
\newcommand{\csplus}{CS_+}
\newcommand{\csminus}{CS_-}
\newcommand{\jfunc}{\mathcal{J}}
\newcommand{\wfunc}{\mathcal{W}}
\newcommand{\cutoff}{\theta_1}
\newcommand{\uth}{\underline{\theta}}
\newcommand{\oth}{\bar{\theta}}

\newcommand{\tth}{\tilde{\theta}}
\newcommand{\qmir}{q^{\mathrm{m}}}

\newcommand{\qL}{\sug{q_L}}
\newcommand{\qH}{\sug{q_H}}

\newcommand{\flatq}{\sug{\bar{q}}}

\usepackage{etoolbox}
\newcounter{eqserial}
\pretocmd{\equation}{%
	\stepcounter{eqserial}%
}{}{\PackageWarning{fix}{no se pudo parchear equation}}

\AtBeginEnvironment{table}{\singlespacing}
\AtBeginEnvironment{figure}{\singlespacing}

\begin{document}

\onehalfspacing
\onehalfspacing

\title{Monotone Allocations without Single-Crossing:\\ When to Bunch and When to Jump\thanks{We thank Humberto Moreira for suggesting that the paper introduce its results through a guiding application, and for raising the question of stochastic contracts, answered in Corollary~\ref{prop_stochastic}. Replication code, the linear-curves testbed of forty configurations with its figures and exact solutions, and an audit of every numerical check are deposited at \url{https://doi.org/10.5281/zenodo.23271020}, the version fixed at submission.}}
\author{A.~Araujo\thanks{IMPA and EPGE/FGV, Rio de Janeiro, Brazil. \texttt{aloisio@impa.br}}
\and C.~Parra\thanks{\raggedright Faculdade de Ci\^encias Econ\^omicas, UERJ, Rio de Janeiro, Brazil. \texttt{carolina.martinez@uerj.br}}
\and S.~Vieira\thanks{Ibmec, Rio de Janeiro, Brazil. \texttt{SVieira2@ibmec.edu.br}}}
\date{This version: October 2026}
\hypersetup{pageanchor=false}\maketitle

\begin{abstract}
We study one-dimensional screening when the Spence--Mirrlees condition
fails along a monotone dividing curve, on which all types value marginal
output equally. A single ratio, marginal distortion over marginal rent,
governs the problem. When the relaxed solution stays in the positive
region, non-decreasing optima in the positive single-crossing region are
continuous if the ratio increases in the decision. When it crosses a flat
curve and marginal rents are separable, every continuous contract is
improved by a jump. When it crosses a strictly monotone curve, a contract
that pools and jumps continues along the mirror of the pool if the ratio
does not increase along it. Optimality among all implementable
allocations, deterministic or random, is proved by a measure on the
binding constraints whose weights are the same ratio; \contrib{with
constant curvatures and no crossing, it is also necessary.} As
applications, we solve problems of procurement and nonlinear pricing.
\end{abstract}
\thispagestyle{empty}\newpage\setcounter{page}{1}\hypersetup{pageanchor=true}

\section{Introduction}\label{sec:Intro}

We study principal--agent problems with adverse selection in which the
Spence--Mirrlees condition (SMC) fails. Under the SMC the indifference
curves of two types cross at most once, global incentive compatibility
reduces to monotonicity of the decision, and the optimal contract is
obtained by the ironing procedure of \citet{MyR} and \citet{GL1984}. The
condition fails in economically natural situations. Consider a buyer who
procures from a supplier of unknown scale. The supplier's cost has two
channels: overhead, which is heavier for larger suppliers, and
congestion, which larger suppliers absorb more easily. Below the scale at
which the two balance a larger supplier has the higher marginal cost;
above it, the lower. This is the procurement example with which
\citet{schottmuller2015} illustrates his algorithm. A contract that
improves on the one reported there is exhibited in \citet{AVP2022},
without a proof of optimality, so the optimum remained open. The same structure arises whenever a
technology has a minimum efficient scale, and whenever several
dimensions of private information are aggregated into one index, as in
the nonlinear pricing problem of \citet{MaskinEtAl1987}. In all these
cases the cross derivative $v_{q\theta}$ changes sign along a
\emph{dividing curve}, the locus where all types value marginal output
equally, and the type--decision space splits into a region of positive
and a region of negative single crossing.\footnote{The same object is
the dividing line of \citet{ChenIshidaSuen2022}; we say curve because in
our applications it is not affine.}

\citet{AM2010} solve the problem when the dividing curve is decreasing.
In that case the binding constraints pair types discretely, through a
U-condition, and the optimal contract can pool isolated pairs and jump.
\citet{AMV2015} develop the marginal tariff approach for that class and
show that discontinuous contracts can dominate continuous ones.
\citet{schottmuller2015} treats the case of two monotone curves, which is
the one produced by index aggregation, characterizes its strictly
monotone continuous solutions, and gives an algorithm to compute them;
his analysis is the starting point of ours. One step in his sufficient
conditions for continuity requires an additional argument, as noted in
\citet{AVP2022}.
Three things are therefore missing for the monotone class: a proof of
continuity where continuity holds, a description of the cases in which
a jump occurs instead, and sufficient conditions that prove, against
every implementable allocation, that the candidates obtained from
first-order conditions are optimal. This paper supplies the three under
explicit conditions on the primitives; the continuity result also
restricts the class of solutions, and the certificate removes that
restriction.

The economic mechanism is the following. Near the dividing curve the
informational rent of every type is flat to second order in the
decision, since $v_{q\theta}=0$ there, while the virtual surplus is
steep to first order. Bunching therefore appears from the erosion of
rents near the curve, and not from a non-monotone relaxed solution as
in the classical theory; and a jump across the curve costs almost
nothing in rents and is strictly profitable. The object that measures
this trade-off is the \emph{distortion--rent ratio}, the loss of virtual
surplus from raising the decision per unit of marginal rent created. Its
monotonicity determines the shape of the optimal contract, within the
class of non-decreasing contracts that stay on one side of the dividing
curve when the curves do not meet; and its value along a contract is
the weight that proves the contract optimal, without any restriction on
the competitors.

The mathematical problem is to maximize the
expected virtual surplus over the allocations that satisfy a continuum
of global incentive constraints, one for each ordered pair of types.
Without single crossing these constraints do not reduce to
monotonicity, and the problem has four parts: shape
(Section~\ref{sec:continuity}), sufficiency
(Section~\ref{sec:sufficiency}), and the binding constraints and the
characterization of the candidates (Section~\ref{sec:characterization}).
\begin{enumerate}[(1)]
\item \emph{Shape.} We prove three necessary conditions by
perturbation. If the relaxed solution does not meet the dividing curve,
every non-decreasing optimum in the positive single-crossing region is
continuous, provided the ratio increases in the decision above the
relaxed solution (Theorem~\ref{proposicion}); this is the continuity
result of \citet{schottmuller2015}, with a condition under which the
argument goes through, and the condition is sufficient and not necessary
(Supplemental Example~\ref{ex:H_fails}). The theorem is stated for the class of non-decreasing allocations
whose graph lies in the positive single-crossing region, the class
studied by \citet{schottmuller2015}; it does not say that the optimum
belongs to that class. The restriction is removed by the certificate of
part (3), which compares the continuous candidate with every
implementable allocation, monotone or not; this is how the
procurement application is solved. If the relaxed solution crosses a flat
dividing curve and marginal rents are separable, every continuous
contract is strictly improved by a jump, under a condition on the
best pool (Proposition~\ref{prop_jump_forced}). If it crosses a strictly monotone
dividing curve, a contract that pools and then jumps must continue
along the \emph{mirror} of the pooled level, the decision on the other
side of the curve with the same marginal rent, provided the ratio does
not increase along the mirror (Proposition~\ref{prop_mirror_shape}).
\item \emph{Global constraints.} We identify, in each configuration,
the family of global constraints that binds at the optimum: the pooled
types against the top type when the curves do not meet, the lower flat
against the upper flat across a horizontal curve, the pool against the
whole mirror across a monotone curve, and the branch against an extreme
type in a fourth family. In each case the binding constraints hold with
equality on an interval, and this determines the non-constant piece of
the contract.
\item \emph{Sufficiency.} A candidate is proved optimal among all
implementable allocations, deterministic or random, by a
weak-duality certificate (Theorem~\ref{prop_suff}): a non-negative
measure on the binding pairs under which the collected pointwise density
is concave with the candidate as its maximizer. The argument is
Luenberger's sufficiency theorem, with multipliers constructed so that
the Lagrangian is concave at them, as in
\citet{AmadorWerningAngeletos2006} and \citet{AmadorBagwell2013}. The
weights of the measure are the distortion--rent ratio along the
contract, and the conditions on them are conditions on the primitives.
\item \emph{Characterization.} In each configuration the candidate is
found by first-order conditions in a reduced family with at most three
free parameters. For the no-crossing case the system is Proposition~1
of \citet{AVP2022}. The rest is new. Across a horizontal dividing curve
the boundary terms of the variation do not cancel at the jump, and the
system acquires a jump condition: the virtual surplus is continuous
across the jump. Across a monotone dividing curve the mirror satisfies
the global constraint by construction, so the system has no
multiplier, and the problem is an unconstrained maximization in two
variables. A fourth family, a pool followed by a branch on which one
extreme type is indifferent to every type of the branch, contains
optima that lie outside the three families suggested by the shape
results. Implementability of the candidates is not automatic: it is a
condition on the constraints anchored at the cutoff and on the branch,
which we state in a form that is necessary and sufficient and that
reduces, on the family with affine curves, to one polynomial
inequality. A reduction lemma shows that for every solution of these
systems the identities the certificate requires, the mass of the
measure at the cutoff and its boundary values, hold automatically, so
the certificate reduces to sign and monotonicity conditions on the
weights; and a Farkas alternative detects candidates outside the four
families. \contrib{Finally, the certificate is not only sufficient. In
the no-crossing structure its concavity condition is necessary: when it
fails, the candidate minimizes the Lagrangian among the allocations
with the same pool. On the family with constant curvatures, which
plays here the role of the family on which \citet{AmadorBagwell2013}
show their conditions to be necessary, the other conditions follow
from the first-order conditions, so the certificate is a
characterization and a single sign decides; in the two crossing
structures every stationary candidate is certified.}
\end{enumerate}

We solve the procurement problem of
\citet{schottmuller2015}: the contract exhibited in \citet{AVP2022} is
the one selected by the first-order system, and the certificate proves
it optimal among all implementable allocations. We also solve a
nonlinear pricing problem with a nonlinear relaxed solution. In a family
with affine curves that realizes the forty configurations of the class,
the certificate decides cases that inspection and grid search do not,
and it overturns two solutions that passed every numerical check. We
also record where the certificate does not apply: one configuration of
the family remains without certificate, and an example with parallel
curves and a surplus that loses curvature admits no certificate of this
form (Supplemental Example~\ref{ex:case2}).

Section~\ref{sec:Modelo} presents the model,
two necessary conditions, and a guiding family whose four regimes
solve in closed form on stated domains. Section~\ref{sec:continuity} contains the three
shape results. Section~\ref{sec:sufficiency} states the certificate,
extends it to random mechanisms and records its limits.
Section~\ref{sec:characterization} characterizes the candidates,
constructs the certificate for each binding structure, shows where it
is exact, and gives the reduction that makes it automatic.
Section~\ref{sec:applications} contains the applications and the closed
forms of the guiding family. Section~\ref{sec:conclusion} reports the
forty configurations and concludes. The proofs of
Theorem~\ref{proposicion}, Proposition~\ref{prop_mirror_shape} and
Theorem~\ref{prop_suff} are in Appendix~\ref{app:proofs_general}. The
proof of Proposition~\ref{prop_jump_forced}, the constructions behind the
certificates, the decreasing configurations, the corner certificates and
the linear-curves family are in the Supplemental Appendix.

\section{Model}\label{sec:Modelo}

A principal contracts with an agent whose preferences are quasi-linear,
$V(q,\theta,t)=v(q,\theta)-t$, where $q\in Q\subset\mathbb R_+$, with
$Q$ a compact interval, is the decision variable and $t\in\mathbb R$ the
monetary transfer. The agent's type $\theta\in\Theta=[\uth,\oth]$ is
private information, drawn from a distribution with cdf $\cdf(\theta)$
and continuous, strictly positive density $\dens(\theta)$. The
principal's payoff is $\profit(q,t,\theta)=B(q,\theta)+t$, where
$B(q,\theta)$ is the principal's benefit net of any type-dependent cost.
Throughout, an \emph{allocation} is a measurable function
$q:\Theta\to Q$, and every statement of optimality is among such
allocations.

\begin{assumption}\label{S1}
$v(q,\theta),\,B(q,\theta)\in C^3$ and $v_\theta>0$ on $Q\times\Theta$.
The agent's outside option is zero.
\end{assumption}

Under $v_\theta>0$ the agent's rent is increasing in $\theta$ and
participation binds only at $\uth$. The condition is not essential. When
the rent has an interior minimum, the reduced problem below is a
relaxation in which participation is imposed at $\uth$ only, and a
solution of the relaxation is optimal for the original problem if its
rent $\int_{\uth}^{\theta}v_\theta(q(t),t)\,dt$ is non-negative for every
type. We check this in the procurement application of
Section~\ref{sec:applications}. Non-monotone rents also appear in
\citet{Jullien2000}, from type-dependent reservation utilities.

\subsection*{The principal's problem}

By the Revelation Principle \citep{Mye1979} we restrict attention to
direct mechanisms $(q,t):\Theta\to Q\times\mathbb R$, and write the
agent's \emph{informational rent} as $U(\theta\,;q):=v(q(\theta),\theta)-t(\theta)$.
The principal solves
\begin{equation}
\tag{$\profit$}\label{maxi}
\max_{\{q(\cdot),\,t(\cdot)\}}\int_{\Theta}\profit(q(\theta),t(\theta),\theta)\,\dens(\theta)\,d\theta,
\end{equation}
subject to \emph{individual rationality} and \emph{incentive compatibility},
\begin{align}
U(\theta\,;q)&\ge0, &&\forall\,\theta\in\Theta, \tag{IR}\label{eq:IR}\\
U(\theta\,;q)&\ge v(q(\theta'),\theta)-t(\theta'), &&\forall\,\theta,\theta'\in\Theta. \tag{IC}\label{eq:IC}
\end{align}
The rent $U(\cdot\,;q)$ is Lipschitz with $U'(\theta)=v_\theta(q(\theta),\theta)$
a.e.\ \citep{MilgromSegal2002,AMV2015}, which together with $U(\uth\,;q)=0$
isolates the transfer implementing $q(\cdot)$. Substituting it into
\eqref{maxi} and integrating by parts gives the \emph{reduced problem}
\begin{equation}
\tag{$\profit_2$}\label{P2}
\max_{q(\cdot)}\int_{\uth}^{\oth}\vs(q(\theta),\theta)\,\dens(\theta)\,d\theta
\quad\text{s.t.\ $q(\cdot)$ implementable,}
\end{equation}
where $q(\cdot)$ is \emph{implementable} if some $t(\cdot)$ makes $(q,t)$
satisfy \eqref{eq:IC}, and the \emph{virtual surplus} is
$\vs(q,\theta)=v(q,\theta)+B(q,\theta)-[(1-\cdf(\theta))/\dens(\theta)]\,v_\theta(q,\theta)$.
Dropping implementability gives the \emph{relaxed problem}
$\max_{q(\cdot)}\int_{\uth}^{\oth}\vs(q(\theta),\theta)\dens(\theta)\,d\theta$,
whose interior optimality condition is $\vs_q(q,\theta)=0$. We assume
that the relaxed maximizer exists; at a boundary of $Q$ the one-sided
condition applies.

\begin{assumption}\label{S2}
$\vs(\cdot,\theta)$ is strictly concave for every $\theta\in\Theta$.
\end{assumption}

Under Assumption~\ref{S2} that condition is sufficient and defines a
unique \emph{relaxed solution} $\qrel(\theta)$.

\subsection*{Failure of single crossing}

\begin{definition}[SMC]
The \emph{Spence--Mirrlees condition} holds when $v_{q\theta}$ has
constant sign: $v_{q\theta}>0$ on $Q\times\Theta$ ($\csplus$) or
$v_{q\theta}<0$ ($\csminus$).
\end{definition}

We relax the SMC as follows.

\begin{assumption}\label{S3}
$v_{qq\theta}(q,\theta)$ has constant sign for $q>0$, and
$v_{q\theta\theta}(q,\theta)$ either has constant sign or vanishes
identically. The \emph{dividing curve} $\qlim(\cdot)$, defined
implicitly by $v_{q\theta}(q,\theta)=0$, is accordingly either strictly
monotone or constant.
\end{assumption}

Differentiating $v_{q\theta}(\qlim(\theta),\theta)=0$ gives
$\qlim'=-v_{q\theta\theta}/v_{qq\theta}$, so the sign of
$v_{q\theta\theta}$ fixes the direction of the curve; the case
$v_{q\theta\theta}\equiv0$ is separability of marginal rents,
$v_\theta(q,\theta)=\varphi(q)+\psi(\theta)$. As in \citet{AM2010}, the
curve $\qlim$ divides $\Theta\times Q$ into two regions of single
crossing, $\csplus$ and $\csminus$.

We study monotone solutions. We therefore consider the configurations
in which the relaxed solution is monotone but not implementable.

\begin{assumption}\label{S4}
$\vs_{q\theta}(q,\theta)$ has constant sign for $q>0$.
\end{assumption}

Under Assumptions~\ref{S2}--\ref{S4} an interior $\qrel$ is strictly
monotone, with the direction given by the sign of $\vs_{q\theta}$. The
signs of $v_{qq\theta}$ and $v_{q\theta\theta}$ fix the direction of
$\qlim$ and the side of $\csplus$. \begin{sergei}
Two curves monotone in the same direction can meet twice, so we impose:

\begin{assumption}[Single transverse intersection]\label{S_transverse}
$\qlim-\qrel$ has at most one zero on $\Theta$, and at an interior zero
the two curves cross transversally, $\qlim'\neq\qrel'$.
\end{assumption}

By a \emph{configuration} we mean the discrete data that
Assumptions~\ref{S3}--\ref{S_transverse} leave free: the side of $\qlim$
on which $\csplus$ lies (the sign of $v_{qq\theta}$), the direction of
each curve (the signs of $v_{q\theta\theta}$ and $\vs_{q\theta}$),
whether the curves meet, which lies above when they do not, and which is
steeper at the crossing when they do. There are forty: twenty-eight with
$\qlim$ strictly monotone and twelve with $\qlim$ constant. The
replication deposit realizes each of them by a family with affine curves
and collects their solutions. In the body of the paper the family is
used only as a control, in Section~\ref{sec:conclusion}.
\end{sergei}

\subsection{Two necessary conditions}\label{sec:necesarias}

Global implementability is captured by the \emph{global incentive
function}
\begin{equation}
\gif{\theta_1}{\theta_2}{q}:=\int_{\theta_1}^{\theta_2}\int_{q(\theta_1)}^{q(\tth)}v_{q\theta}(\xi,\tth)\,d\xi\,d\tth,
\tag{GIF}\label{eq:GIF}
\end{equation}
\begin{sergei}
with both integrals oriented. The first argument is the true type and
the second is the imitated contract, as in Definition~2 of
\citet{AVP2022}: \eqref{eq:GIF} is the gain of type $\theta_2$ from its
own contract relative to the contract of $\theta_1$. We call
$\theta_1$, the type whose contract is imitated, the \emph{anchor} of
the constraint, and $\theta_2$, the imitating type, its \emph{endpoint}.
An allocation is implementable if and only if
$\Phi(\theta,\hat\theta;q)\ge0$ for every ordered pair
$(\theta,\hat\theta)$ \citep{AMV2015}.
\end{sergei}

The local content of this condition is the following \citep{AM2010}.

\begin{lemma}[Local Monotonicity, LMC]\label{lmc}
Let $q$ be c\`adl\`ag at $\theta$. If $v_{q\theta}(q(\theta),\theta)>0$
(resp.\ $<0$), then $q$ is non-decreasing (resp.\ non-increasing) on
$(\theta,\theta+\varepsilon)$ for $\varepsilon>0$ small.
\end{lemma}

\begin{sergei}
Applied to $\qrel$, the lemma gives a necessary condition: if it fails at
every type, $\qrel$ is not implementable. Other non-constant allocations
may still be implementable, so full pooling requires a separate
optimality argument (Supplemental Example~\ref{ex:case2}).
\end{sergei}

Under the SMC the local condition is sufficient for implementability.
Without it, global constraints can bind even for monotone decisions
contained in a single region \citep{AM2010,AMV2015}. Substituting the
rent into \eqref{eq:GIF}, the condition for $\theta_1<\theta_2$ takes
the \emph{isoperimetric} form
\begin{equation}
0\le\gif{\theta_1}{\theta_2}{q}=\int_{\theta_1}^{\theta_2}\bigl\{v_\theta(q(\tth),\tth)-v_\theta(q(\theta_1),\tth)\bigr\}\,d\tth.
\label{eq:iso}\tag{ISO($\theta_1,\theta_2$)}
\end{equation}
The \emph{semi-relaxed problem} is \eqref{P2} with implementability
replaced by the single equality \eqref{eq:iso}, for one pair of types.

\begin{lemma}[\cite{AMV2015,Petrov1968}]\label{teo11}
The solution $\qiso(\theta)$ of the semi-relaxed problem satisfies
$\vs_q(\qiso(\theta),\theta)\dens(\theta)+\lambda\,v_{q\theta}(\qiso(\theta),\theta)=0$,
where $\lambda\in\mathbb R$ is chosen so that \eqref{eq:iso} holds with
equality.
\end{lemma}

All proofs are in \appref{app:proofs_general}, except where noted.

\subsection{A guiding family}\label{sec:MES_env}

A principal procures a quantity $q\ge0$ from a firm whose productivity
$\theta\sim U[0,1]$ is private information. The technology has a
minimum efficient scale that depends on the type: below it a higher
type values marginal output less, above it more, and the dividing curve
is the locus where all types value marginal output equally.\footnote{The
name is an analogy: $\qlim$ is the locus where $v_{q\theta}$ changes
sign, not a point of an average cost curve, and for $\alpha<0$ it is
negative on part of the type space.} Gross utility is
\begin{equation}
v(q,\theta)=\frac{\theta q^2}{2}-\Bigl(\alpha+\frac{\gamma}{2}\theta\Bigr)\theta q+q+\theta,
\qquad\alpha\in\mathbb R,\;\gamma\ge0,
\tag{MES}\label{eq:MES_v}
\end{equation}
so that $v_{q\theta}(q,\theta)=q-\alpha-\gamma\theta$ and the dividing
curve is
\begin{equation}
\qlim(\theta)=\alpha+\gamma\theta.
\label{eq:MES_qlim}
\end{equation}
The parameter $\alpha$ is the scale common to all types and $\gamma$ is
the rate at which the efficient scale grows with the type. With the
principal's cost
\begin{equation}
\cost(q,\theta)=\theta q^2+\Bigl(1+\alpha-\theta-2\alpha\theta+\gamma\theta-\frac{3\gamma}{2}\theta^2\Bigr)q+2\theta-1,
\label{eq:MES_C}
\end{equation}
the virtual surplus is $\vs(q,\theta)=-q^2/2+\theta q$ and the relaxed
solution is $\qrel(\theta)=\theta$. The cost is chosen so that the
family solves in closed form; linearity is not a hypothesis of the
results below. The family satisfies Assumptions~\ref{S1}--\ref{S4} when
$\alpha+\gamma<\sqrt2$, which we assume throughout: this bounds $\qlim$
on $[0,1]$ and gives $v_\theta>0$, so \eqref{eq:IR} binds only at
$\theta=0$.\footnote{Assumptions~\ref{S2}
and~\ref{S4} hold unconditionally, since $\vs=-q^2/2+\theta q$
regardless of $(\alpha,\gamma)$. So does Assumption~\ref{S3}, for every
$\gamma\ge0$: $v_{qq\theta}=1$ and $v_{q\theta\theta}=-\gamma$,
vanishing identically at the horizontal corner $\gamma=0$. Only
$v_\theta>0$ binds: $v_\theta(q,\theta)=q^2/2-\qlim(\theta)q+1$ has
minimum $1-\tfrac12\qlim(\theta)^2$ on $\mathbb R_+$ when
$\qlim(\theta)\ge0$ and $1$ otherwise, so $\alpha+\gamma<\sqrt2$
suffices.}

\label{sec:three_regimes}The pair $(\alpha,\gamma)$ selects the configuration. We call the
four cases below the \emph{regimes} of the guiding family, and reserve
\emph{configuration} for the geometric cases of
Assumption~\ref{S_transverse}; the \emph{binding structure} of a
contract, the family of constraints that bind at it, is introduced in
Section~\ref{sec:sufficiency}.
\begin{enumerate}[(i)]
\item \emph{No crossing.} $\alpha=-a$, $a>0$, $\gamma=1$: the dividing
curve $\qlim=\theta-a$ runs parallel below $\qrel=\theta$, which lies
entirely in $\csplus$.
\item \emph{Horizontal crossing.} $\gamma=0$, $0<\alpha<1/\sqrt2$: the
dividing curve is flat, $\qlim\equiv\alpha$, crossed at
$\theta^*=\alpha$.
\item \emph{Monotone crossing, relaxed solution steeper.} $0<\gamma<1$,
$0<\alpha<1-\gamma$: $\qrel$ crosses from $\csminus$ into $\csplus$ at
$\theta^*=\alpha/(1-\gamma)$.
\item \emph{Monotone crossing, dividing curve steeper.} $\gamma>1$,
$(\gamma-1)/2<-\alpha<\gamma-1$: $\qrel$ crosses from $\csplus$ into
$\csminus$ at $\theta^*=-\alpha/(\gamma-1)$.
\end{enumerate}
Section~\ref{sec:MES_solutions} gives the solutions and the parameter
domains on which they are established: the range $a<\tfrac13$ in (i), the whole range in (ii),
the range of (iv) under a sign condition, and a sub-range of (iii).
Figure~\ref{fig:panel_triple} shows the four.\footnote{The family also
contains $\alpha=a>0$, $\gamma=1$, for which $\qrel=\theta<\qlim=a+\theta$
and $\qrel$ violates the LMC everywhere. For these quadratic primitives
the pool $q=1/2$ is optimal for every $a$ in the domain of the family,
with a profit gap of at least $\frac{a}{1+2a}\int(q-\frac12)^2$ over any
other implementable allocation (Supplemental
Proposition~\ref{prop:MES_certificates}). This is a property of the
primitives and not of the geometry: Supplemental Example~\ref{ex:case2}
has the same geometry and an implementable non-constant allocation that
beats every pool.}

\begin{figure}[!tb]
\centering

\newcommand{\mespanel}[2]{%
  \begin{scope}[shift={(#1,#2)}]
    \draw[->] (0,0) -- (1.08,0) node[right,font=\scriptsize] {$\theta$};
    \draw[->] (0,0) -- (0,1.08) node[above,font=\scriptsize] {$q$};
    \draw (0,1) -- (-0.02,1) node[left,font=\scriptsize] {$1$};
  \end{scope}}
\begin{tikzpicture}[x=3.9cm,y=3.2cm,>=Stealth,font=\small]
  \begin{scope}
    \fill[black!7] (0.25,0) -- (1,0.75) -- (1,1.02) -- (0,1.02) -- (0,0) -- cycle;
    \node[black!60,font=\scriptsize] at (0.18,0.9) {$\csplus$};
    \draw[black!35,line width=0.8pt] (0.25,0) -- (1,0.75);
    \node[black!60,font=\scriptsize,anchor=west] at (1.0,0.75) {$\qlim$};
    \draw[black!55,dashed,line width=0.8pt] (0,0) -- (1,1);
    \node[black!60,font=\scriptsize,anchor=west] at (1.0,0.97) {$\qrel$};
    \draw[black,line width=1.3pt] (0,0.1786) -- (0.1429,0.1786) -- (1,1.0357);
    \node[anchor=south,font=\scriptsize] at (0.5,-0.2) {I. No crossing};
  \end{scope}
  \mespanel{0}{0}
  \begin{scope}[shift={(1.35,0)}]
    \fill[black!7] (0,0.5) -- (1,0.5) -- (1,1.02) -- (0,1.02) -- cycle;
    \node[black!60,font=\scriptsize] at (0.18,0.9) {$\csplus$};
    \draw[black!35,line width=0.8pt] (0,0.5) -- (1,0.5);
    \node[black!60,font=\scriptsize,anchor=west] at (1.0,0.5) {$\qlim$};
    \draw[black!55,dashed,line width=0.8pt] (0,0) -- (1,1);
    \node[black!60,font=\scriptsize,anchor=west] at (1.0,0.97) {$\qrel$};
    \draw[black,line width=1.3pt] (0,0.2929) -- (0.5,0.2929);
    \draw[black,dotted,line width=0.8pt] (0.5,0.2929) -- (0.5,0.7071);
    \draw[black,line width=1.3pt] (0.5,0.7071) -- (0.7071,0.7071) -- (1,1);
    \node[anchor=south,font=\scriptsize] at (0.5,-0.2) {II. Horizontal crossing};
  \end{scope}
  \mespanel{1.35}{0}
  \begin{scope}[shift={(0,-1.42)}]
    \fill[black!7] (0,0.375) -- (1,0.625) -- (1,1.02) -- (0,1.02) -- cycle;
    \node[black!60,font=\scriptsize] at (0.18,0.9) {$\csplus$};
    \draw[black!35,line width=0.8pt] (0,0.375) -- (1,0.625);
    \node[black!60,font=\scriptsize,anchor=west] at (1.0,0.625) {$\qlim$};
    \draw[black!55,dashed,line width=0.8pt] (0,0) -- (1,1);
    \node[black!60,font=\scriptsize,anchor=west] at (1.0,0.97) {$\qrel$};
    \draw[black,line width=1.3pt] (0,0.3170) -- (0.5000,0.3170);
    \draw[black,dotted,line width=0.8pt] (0.5000,0.3170) -- (0.5000,0.6830);
    \draw[black,line width=1.3pt] (0.5000,0.6830) -- (0.8660,0.8660) -- (1,1);
    \node[anchor=south,font=\scriptsize] at (0.5,-0.2) {III. Monotone crossing, $\qrel$ steeper};
  \end{scope}
  \mespanel{0}{-1.42}
  \begin{scope}[shift={(1.35,-1.42)}]
    \fill[black!7] (0.2667,0) -- (0.9333,1.02) -- (0,1.02) -- (0,0) -- cycle;
    \node[black!60,font=\scriptsize] at (0.18,0.9) {$\csplus$};
    \draw[black!35,line width=0.8pt] (0.2667,0) -- (0.9333,1.0);
    \node[black!60,font=\scriptsize,anchor=south west] at (0.9,0.98) {$\qlim$};
    \draw[black!55,dashed,line width=0.8pt] (0,0) -- (1,1);
    \node[black!60,font=\scriptsize,anchor=west] at (1.0,0.86) {$\qrel$};
    \draw[black,line width=1.3pt] (0,0.3500) -- (0.6000,0.8000) -- (1,0.8000);
    \node[anchor=south,font=\scriptsize] at (0.5,-0.2) {IV. Monotone crossing, $\qlim$ steeper};
  \end{scope}
  \mespanel{1.35}{-1.42}
\end{tikzpicture}
\caption{The four regimes of the guiding family, at
$(\alpha,\gamma)=(-\tfrac14,1)$, $(\tfrac12,0)$, $(\tfrac38,\tfrac14)$ and
$(-\tfrac25,\tfrac32)$: no crossing (continuous, pool then branch),
horizontal crossing (two flats and a jump), monotone crossing (pool,
jump onto the mirror, relaxed tail), and a dividing curve steeper than
the relaxed solution (branch on the indifference locus of the top type,
then a pool at the crossing height). Dividing
curve $\qlim$ in gray, relaxed solution $\qrel$ dashed, $\csplus$
shaded, contract solid. Dotted vertical segments are jumps and not
offered quantities. The four contracts are proved optimal in
Supplemental Proposition~\ref{prop:MES_certificates}.}
\label{fig:panel_triple}
\end{figure}

\FloatBarrier

\section{Continuity and Jumps}\label{sec:continuity}\label{sec:MES}

Without single crossing the marginal value of the decision to a type
changes sign along $\qlim(\theta)$: it is positive above the curve and
negative below. This section gives three results on the shape of the
optimum: when it is continuous, when it jumps, and what a jump must
look like. Bunching and jumps are different features and can appear in
the same contract. Each result has a hypothesis on the primitives and a
class of allocations in which it applies; Section~\ref{sec:scope} states
what the results leave open. The three are necessary conditions obtained
from perturbations (smoothing a jump, opening a jump, lowering a
branch), as in \citet[\S3.4]{AmadorBagwell2013}. Sufficiency is treated
in Section~\ref{sec:sufficiency}. Throughout, a \emph{pool} is a
maximal interval on which the contract is constant, its right end is
the \emph{cutoff}, and a \emph{branch} is a maximal interval on which
the contract is strictly increasing.

\subsection{Continuity without crossing}\label{sec:shape_continuity}

Suppose that $\qrel$ lies in one single-crossing region. Sufficient
conditions for continuity of the optimum are given in
\citet[Propositions~3 and~4]{schottmuller2015}. Their proofs rely on a step
that, as noted in \citet{AVP2022}, requires an additional argument, and
we do not know one that dispenses with an extra condition.\footnote{Both use the sixth bullet point of
Theorem~1 in \citet{schottmuller2015}.} Theorem~\ref{proposicion} below
uses an additional condition on the primitives, Assumption~H, and holds
in a stated class of allocations.

The condition is motivated as follows. Suppose the
optimum jumps and we smooth the jump over a band of types to its right.
Raising the decision of the types just below the jump is costly and
lowering the decision of the types just above is profitable, so the
band must be long for the change to be profitable; incentive
compatibility across the band requires it to be short.

\begin{assumptionH}\label{ass:H}
For every $\theta\in\Theta$, the distortion--rent ratio
\begin{equation}\label{eq:ratio}
R(q,\theta):=-\frac{\vs_q(q,\theta)}{v_{q\theta}(q,\theta)}
\end{equation}
is strictly increasing in $q$ on $\{q:q\ge\qrel(\theta)\}$.
\end{assumptionH}

$R$ is well defined where it is used: $q\ge\qrel(\theta)>\qlim(\theta)$
gives $v_{q\theta}>0$, and $\vs_q\le0$ above $\qrel$ gives $R\ge0$. Up
to the density, $R$ is the ratio that \citet[eq.~(11)]{AM2010} equalize
across discretely pooled types: the marginal profit cost of raising the
decision per unit of marginal rent created. Assumption~H says that this
ratio increases along the fiber above $\qrel$. It is a condition on the
primitives, independent of Assumptions~\ref{S1}--\ref{S4}, and it can
be checked before solving the problem. In the family with affine curves
it holds whenever $\qrel\subset\csplus$, since
$\partial_qR=-\vs_{qq}\,v_{q\theta}(\qrel(\theta),\theta)/v_{q\theta}^2>0$
there; in the guiding family it therefore holds in regime~(i). It also
holds in the procurement application of Section~\ref{sec:applications},
whose primitives are nonlinear. The same ratio appears in
Section~\ref{sec:sufficiency} as the shadow price of a global
constraint.

\begin{sergei}
\begin{theorem}\label{proposicion}
Let $\qlim$ be increasing with $\csplus$ above, $\qrel$ increasing with
graph contained in $\csplus$, and let Assumption~H hold. Let
$\mathcal C$ be the set of implementable, non-decreasing, c\`adl\`ag
allocations with graph in $\csplus$. If $q\in\mathcal C$ maximizes
\eqref{maxi} over $\mathcal C$, then $q$ is continuous on $[\uth,\oth)$,
and on $\Theta$ when $q(\oth)=\lim_{t\uparrow\oth}q(t)$.
\end{theorem}
\end{sergei}

The theorem applies in particular to an optimal solution of
\eqref{maxi} that belongs to $\mathcal C$. The proof, in
Appendix~\ref{app:shape}, is by contradiction. If the optimum jumps from
$l$ to $u$ at a type $\tau$, replace it by an intermediate level $z$ on a
band of types around $\tau$, of length $\varepsilon$ to the left and
$L\varepsilon$ to the right. On the left the decision rises, which costs
profit and creates rent; on the right it falls, which gains profit and
destroys rent. Assumption~H says that the distortion--rent ratio is
larger on $[z,u]$ than on $[l,z]$, so $L$ can be chosen so that the
profit gained exceeds the profit lost and the rent created exceeds the
rent destroyed. The constraints of the types in the band hold because
their rent is concave in their decision. Hence no jump is optimal.
Figure~\ref{fig:smoothing} shows the construction.

\begin{figure}[!tb]
\centering
\begin{tikzpicture}[x=9cm,y=5.2cm,>=Stealth,font=\small]
  % CS+ above q0
  \fill[black!6] (0,0.10) -- (1,0.42) -- (1,1.02) -- (0,1.02) -- cycle;
  \draw[black!35,line width=0.8pt] (0,0.10) -- (1,0.42) node[right,black!60] {$q_0$};
  \draw[black!55,dashed,line width=0.8pt] (0,0.22) -- (1,0.60) node[right,black!60] {$q_1$};
  % the two bands
  \fill[black!30] (0.40,0.50) -- (0.50,0.555) -- (0.50,0.62) -- (0.40,0.62) -- cycle;
  \fill[black!15] (0.50,0.62) -- (0.50,0.80) -- (0.62,0.84) -- (0.62,0.62) -- cycle;
  % the contract with the jump
  \draw[black,line width=1.4pt] (0,0.30) -- (0.50,0.555);
  \draw[black,dotted,line width=0.8pt] (0.50,0.555) -- (0.50,0.80);
  \draw[black,line width=1.4pt] (0.50,0.80) -- (1,0.97) node[right] {$q$};
  % the smoothed level z
  \draw[black,line width=1.1pt,dashed] (0.40,0.62) -- (0.62,0.62);
  \node[anchor=east] at (0.40,0.62) {$z$};
  \node[anchor=west,font=\footnotesize] at (0.505,0.53) {$l$};
  \node[anchor=east,font=\footnotesize] at (0.495,0.81) {$u$};
  % labels of the two regions
  \node[align=center,font=\footnotesize] at (0.18,0.80) {decision rises:\\ rent created,\\ profit lost};
  \draw[black!60,->] (0.27,0.74) -- (0.44,0.595);
  \node[align=center,font=\footnotesize] at (0.87,0.72) {decision falls:\\ rent destroyed,\\ profit gained};
  \draw[black!60,->] (0.775,0.71) -- (0.59,0.71);
  % axes and ticks
  \draw[->] (0,0) -- (1.03,0) node[right] {$\theta$};
  \draw[->] (0,0) -- (0,1.05) node[above] {$q$};
  \draw (0.40,0) -- (0.40,-0.015) node[below] {$\tau-\varepsilon$};
  \draw (0.50,0) -- (0.50,-0.015) node[below] {$\tau$};
  \draw (0.62,0) -- (0.62,-0.015) node[below] {$\tau+L\varepsilon$};
\end{tikzpicture}
\caption{Smoothing a jump. The contract $q$ jumps from $l$ to $u$ at
$\tau$. On $[\tau-\varepsilon,\tau+L\varepsilon)$ it is replaced by the
level $z$. By Assumption~H the distortion--rent ratio is larger on
$[z,u]$ than on $[l,z]$, so the band length $L$ can be chosen so that the
profit gained on the right exceeds the profit lost on the left while
the rent created on the left exceeds the rent destroyed on the right.}
\label{fig:smoothing}
\end{figure}

\begin{sergei}
\begin{remark}[The class $\mathcal C$]\label{rem:classC}
By Lemma~\ref{lmc}, an implementable non-decreasing allocation is
constant wherever it lies below $\qlim$, so it leaves $\csplus$ only on
a pool that $\qlim$ overtakes. Hence $\mathcal C$ is the class of
implementable non-decreasing c\`adl\`ag allocations none of whose pools
is overtaken by $\qlim$; since $\qlim$ is increasing, the condition is
$q(b^-)\ge\qlim(b)$ at the right end $b$ of every pool. In particular
every implementable strictly increasing allocation belongs to
$\mathcal C$ (Corollary~1 of \citet{schottmuller2015}), and the class
excludes exactly one shape: a pool overtaken by $\qlim$ followed by a
jump back into $\csplus$. The theorem does not say that an optimum of
\eqref{maxi} cannot have that shape when $\qrel\subset\csplus$; this is
the question it leaves open. Where the certificate of
Section~\ref{sec:sufficiency} applies the question does not arise,
since the certificate compares the continuous candidate with every
implementable allocation, monotone or not; this is how the
procurement application proceeds.
\end{remark}
\end{sergei}

Assumption~H supplies the step of the argument of
\citet{schottmuller2015} noted in \citet{AVP2022}: it makes a short
band strictly profitable. If H fails the argument fails, but we do not know whether
the conclusion fails: we know no jump permitted by the failure of H and
no proof that none exists. We keep H because it closes the argument, it
can be checked before solving, and, as Supplemental
Example~\ref{ex:H_fails} shows, it is sufficient and need not be necessary:
there H fails on the fiber above $\qrel$, at decisions above those the
contract uses, and the continuous candidate passes every clause of the
certificate, by quadrature and on grids (Supplemental Part~II): numerical
validation at the displayed calibrations. The
certificate reaches exactly as far as H does: on the branch its concavity condition is
H evaluated on the contract, so where H fails on the contract itself the
certificate does not apply. In that range the same example, computed
with every global incentive constraint, still has a continuous optimum;
that is numerical evidence, and what sustains continuity there is open.

\subsection{Jumps with a horizontal dividing curve}\label{sec:shape_jump}

When the dividing curve is flat the conclusion is reversed: every
continuous contract is strictly improved by opening a jump. The reason
is that $v_{q\theta}(\qlim,\cdot)=0$. Informational rents are flat to
second order at the dividing curve, while the virtual surplus is steep
to first order, so a discontinuity across $\qlim$ costs almost nothing
in rents and is strictly profitable.

Suppose $v_{q\theta\theta}\equiv0$, so that marginal rents are
separable, $v_\theta(q,\theta)=\varphi(q)+\psi(\theta)$. If $\varphi$
is strictly convex with minimum at $c$, then $\qlim\equiv c$, and
quantities on opposite sides of $c$ can generate the same marginal
rent. The candidate family $D_J$ of Section~\ref{sec:spaces} consists
of two flat levels separated by a single jump, followed by the relaxed
solution. The proposition below excludes continuity; it does not say
that $D_J$ contains every optimum.

\begin{proposition}[Unavoidable jump: horizontal case]\label{prop_jump_forced}
Let $v_{q\theta\theta}\equiv0$, so that $v_\theta(q,\theta)=\varphi(q)+\psi(\theta)$
with $\varphi$ strictly convex, $\varphi''(c)>0$, minimized at $c$, and
$\qlim\equiv c$. Assume $Q$ is an interval, $c$ is interior to $Q$, and
the mirror used below belongs to $Q$. Let $\qrel$ be strictly
increasing with $\qrel(\theta^*)=c$ for some $\theta^*\in(\uth,\oth)$,
and let $\flatq^*$ be the pooled level, the unique root of
$\int_\Theta\vs_q(\flatq,\theta)\dens(\theta)\,d\theta=0$. If
$\flatq^*<c$, write $\qmir(\flatq^*)>c$ for its mirror,
$\varphi(\qmir)=\varphi(\flatq^*)$, and assume
\begin{equation}\label{eq:jump_condition}
\vs\bigl(\qmir(\flatq^*),\oth\bigr)>\vs(\flatq^*,\oth).
\end{equation}
Then no continuous allocation solves \eqref{maxi}: every continuous
implementable allocation is strictly dominated by an implementable
allocation with two flat pieces separated by a single jump and a
non-decreasing tail.\footnote{At a jump type $\theta_1$ the rent is
absolutely continuous, so $\theta_1$ is indifferent between the two
levels of the jump. As in \citet{AM2010}, we read the decision at
$\theta_1$ as the interval $[\qL,\qH]$, a convex-valued correspondence
whose c\`adl\`ag selection is the allocation; the intermediate
quantities are not offered.}
\end{proposition}

\begin{sergei}
The proof is in Section~\ref{sec:proofs_main} of the Supplemental
Appendix. It opens a jump of size $\delta$ around $c$, with the two
flats placed so that the marginal type remains indifferent. The gain is
of first order in $\delta$ because $\vs_q(c,\cdot)\neq0$ on both sides,
and the cost in rents is of second order because $v_{q\theta}(c,\cdot)=0$.
Condition \eqref{eq:jump_condition} compares, for the top type, the
mirror of the best pool with the pool; it holds when
$\qmir(\flatq^*)\le\qrel(\oth)$, and when $\vs$ is quadratic in $q$ and
$\varphi$ is symmetric about $c$, where it reduces to $c<\qrel(\oth)$,
the existence of the crossing. It fails only when the mirror of the
best pool is far above $\qrel$, and then a pool below $c$ can be
optimal. The same argument applies to $\qrel$ strictly decreasing, with
the jump downward. The proposition excludes continuity; the location
and the number of jumps, and whether two flats and a tail are the only
possible optima, are questions for the variational system of
Section~\ref{sec:characterization} and for the certificate.
\end{sergei}

\begin{newresult}

\subsection{The mirror}
\label{sec:shape_mirror_main}

We now consider $\qlim$ strictly increasing with $\csplus$ above,
$\qrel$ increasing, and $\qrel$ crossing $\qlim$ once from $\csminus$
into $\csplus$. Neither Theorem~\ref{proposicion} nor
Proposition~\ref{prop_jump_forced} applies. Throughout this subsection
we take $v_{qq\theta}>0$ and $v_{q\theta\theta}\le0$ in
Assumption~\ref{S3}.

The situation is the following. Below the crossing the relaxed
solution lies in $\csminus$, where higher types value marginal output
less. There an implementable allocation is non-increasing
(Lemma~\ref{lmc}) while $\qrel$ is increasing, so the contract cannot
follow $\qrel$ and pools the low types at some level $\flatq$. Above the
crossing $\qrel$ lies in $\csplus$ and can be followed. To go from the
pool to $\csplus$ the contract has to cross $\qlim$, and by the
argument of Section~\ref{sec:shape_jump} it crosses with a jump. Two
questions remain: where the jump lands, and how the contract continues
after it. Both are answered by the constraints of the pooled types
against the types above them.

Fix a cutoff $\cutoff\in(\uth,\oth)$ and a pooled level $\flatq$. For
a type $\theta$ with $\flatq<\qlim(\theta)$ and a decision $z$, let
\begin{equation}\label{eq:h_main}
h(z,\theta):=v_\theta(z,\theta)-v_\theta(\flatq,\theta)
=\int_{\flatq}^{z}v_{q\theta}(\xi,\theta)\,d\xi
\end{equation}
be the marginal rent of $z$ in excess of the marginal rent of the pool.
Since $v_{q\theta}(\cdot,\theta)$ is increasing and vanishes at
$\qlim(\theta)$, $h(\cdot,\theta)$ decreases on $(\flatq,\qlim(\theta))$
and increases on $(\qlim(\theta),\infty)$, with $h(\flatq,\theta)=0$ and
$h(\qlim(\theta),\theta)<0$. Hence there is exactly one decision above
the curve with the same marginal rent as the pool.

\begin{definition}[Mirror]\label{def:mirror_main}
For $\theta$ with $\flatq<\qlim(\theta)$, the \emph{mirror} of $\flatq$ at
$\theta$ is the unique $\qmir(\theta)>\qlim(\theta)$ with
$h(\qmir(\theta),\theta)=0$, that is,
$v_\theta(\qmir(\theta),\theta)=v_\theta(\flatq,\theta)$. We assume it
exists in $Q$ for every such $\theta$. Differentiating,
\[
\qmir{}'(\theta)=-\int_{\flatq}^{\qmir(\theta)}v_{q\theta\theta}(\xi,\theta)\,d\xi
\Big/v_{q\theta}(\qmir(\theta),\theta)\ge0,
\]
so the mirror is non-decreasing. When $v_{q\theta}(\cdot,\theta)$ is affine,
$\qmir=2\qlim-\flatq$.
\end{definition}

The mirror is the decision on the other side of the curve that leaves
the marginal rent of type $\theta$ unchanged. Consider an allocation $q$
that pools at $\flatq$ on $[\uth,\cutoff)$. For an anchor $a<\cutoff$ in
the pool and a type $e\ge\cutoff$,
\begin{equation}\label{eq:G_main}
\gif{a}{e}{q}=\int_{\cutoff}^{e}h(q(x),x)\,dx=:G_q(e),
\end{equation}
which does not depend on $a$: the constraint of a pooled type against a
type $e$ above the cutoff is the excess marginal rent accumulated by the
contract between the cutoff and $e$, and implementability requires
$G_q\ge0$ on $[\cutoff,\oth]$. This gives the intuition for the two
questions. Where the contract runs between $\flatq$ and the mirror,
$h<0$ and the constraint loses slack; where it runs above the mirror,
$h>0$ and the constraint gains slack. A contract that lands below the
mirror makes $G_q$ negative immediately after the jump, so the jump
lands on the mirror or above it. A contract that follows the mirror
keeps $G_q=0$: every pooled type is exactly indifferent to every branch
type. A contract above the mirror is implementable but distorts more,
since the mirror is closer to $\qrel$; it accumulates slack that can be
used later to run below the mirror, where the distortion is smaller.
Whether this is profitable depends on the distortion--rent ratio along
the mirror.

Figure~\ref{fig:mirror} shows the geometry.

\begin{figure}[!tb]
\centering
\begin{tikzpicture}[x=5.6cm,y=3.9cm,>=Stealth,font=\small]
  \fill[black!7] (0,0.21) -- (1,0.51) -- (1,1.02) -- (0,1.02) -- cycle;
  \fill[black!25] (0.3,0.1955) -- (0.561,0.1955) -- (0.561,0.561) -- (0.3,0.4045) -- cycle;
  \draw[black!35,line width=0.8pt] (0,0.21) -- (1,0.51) node[right,black!60] {$q_0$};
  \draw[black!55,dashed,line width=0.8pt] (0,0) -- (1,1) node[right,black!60] {$q_1$};
  \draw[black!55,dotted,line width=0.6pt] (0.561,0.561) -- (1,0.8245);
  \draw[black,line width=1.5pt] (0,0.1955) -- (0.3,0.1955);
  \draw[black,dotted,line width=0.8pt] (0.3,0.1955) -- (0.3,0.4045);
  \draw[black,line width=1.5pt] (0.3,0.4045) -- (0.561,0.561);
  \draw[black,line width=1.5pt] (0.561,0.561) -- (1,1);
  \node[black!60] at (0.12,0.72) {$CS_+$};
  \node[black!60] at (0.82,0.25) {$CS_-$};
  \node[anchor=south west] at (0.01,0.1955) {$\bar q$};
  \node[anchor=south east] at (0.5,0.53) {$q^{\mathrm m}$};
  \node[font=\footnotesize] at (0.43,0.31) {$h<0$};
  \draw[->] (0,0) -- (1.03,0) node[right] {$\theta$};
  \draw[->] (0,0) -- (0,1.03) node[above] {$q$};
  \draw (0.3,0) -- (0.3,-0.015) node[below] {$\theta_1$};
  \draw (0.561,0) -- (0.561,-0.015) node[below] {$\theta_3$};
\end{tikzpicture}
\caption{The mirror. The contract (solid) pools at $\flatq$ below
$\theta_1$, jumps to the mirror $\qmir$ of $\flatq$, follows it until
$\theta_3$, where $\qrel$ (dashed) reaches it, and follows $\qrel$
afterwards. In the shaded region between $\flatq$ and $\qmir$ the excess
marginal rent $h$ is negative, so a contract there violates the
constraints of the pooled types; above $\qmir$, $h$ is positive.
Guiding family, regime III.}
\label{fig:mirror}
\end{figure}

Let $\theta_3$ be the first type in $[\cutoff,\oth]$ with
$\qmir(\theta_3)=\qrel(\theta_3)$, and $\theta_3:=\oth$ if there is none.
On $[\cutoff,\theta_3]$ define the \emph{mirror weight}
\begin{equation}\label{eq:Lambda_mirror}
\Lambda(t):=R\bigl(\qmir(t),t\bigr)\,\dens(t)
=-\frac{\vs_q(\qmir(t),t)\,\dens(t)}{v_{q\theta}(\qmir(t),t)} ,
\end{equation}
the distortion--rent ratio of \eqref{eq:ratio} evaluated along the
mirror. It is well defined and bounded, since $\qmir(t)>\qlim(t)$ keeps
$v_{q\theta}(\qmir(t),t)$ away from zero on $[\cutoff,\theta_3]$. It is
non-negative, since $\qmir\ge\qrel$ there gives $\vs_q(\qmir,\cdot)\le0$,
and $\Lambda(\theta_3)=0$ when $\qmir(\theta_3)=\qrel(\theta_3)$, in
particular when $\theta_3<\oth$.

\begin{assumptionM}\label{ass:M}
Given $(\flatq,\cutoff)$: \textup{(M)} $\Lambda$ is non-increasing on
$[\cutoff,\theta_3]$; \textup{(C)}
$\vs_{qq}(z,t)\dens(t)+\Lambda(t)\,v_{qq\theta}(z,t)<0$ for every
$z\in Q$ and almost every $t\in[\cutoff,\theta_3]$.
\end{assumptionM}

$\Lambda(t)$ is the profit lost per unit of slack gained at type $t$
by moving the contract above the mirror. Under (M) the ratio does not
increase along the mirror, so slack accumulated at a type $t$ is worth
less than the profit it costs at every later type: running above the
mirror to run below it later does not pay. (C) is the concavity
condition that also appears in the certificate of
Section~\ref{sec:sufficiency}. Both are conditions on the primitives,
given the pool and the cutoff.

\begin{proposition}[The branch after the pool is the mirror]
\label{prop_mirror_shape}
Let $\qlim$ be increasing with $\csplus$ above and $\qrel$ increasing.
\begin{enumerate}[(i)]
\item \textup{(Pooling in $\csminus$.)} An implementable, non-decreasing
c\`adl\`ag allocation with $q<\qlim$ on $[\uth,\cutoff)$ is constant
there.
\item \textup{(Landing.)} If in addition $q=\flatq<\qlim(\cutoff)$ on
$[\uth,\cutoff)$ and $q(\cutoff)>\flatq$, then $q(\cutoff)\ge\qmir(\cutoff)$.
\item \textup{(Mirror.)} Let $\flatq<\qlim(\cutoff)$ with
$\qrel(\cutoff)\le\qlim(\cutoff)$, let $\qmir$ and
$\qrel$ meet at most once on $[\cutoff,\oth]$, so that $\qrel\ge\qmir$ on
$[\theta_3,\oth]$, and let Assumption~\ref{ass:M} hold for
$(\flatq,\cutoff)$. Define
\[
r(\theta):=\flatq\ \text{ on }[\uth,\cutoff),\qquad
r(\theta):=\max\{\qmir(\theta),\qrel(\theta)\}\ \text{ on }[\cutoff,\oth].
\]
Then $r$ is implementable and, for every implementable allocation $q$
with $q=\flatq$ on $[\uth,\cutoff)$,
\[
\jfunc(q)\le\jfunc(r),
\qquad\text{with equality only if } q=r \text{ almost everywhere.}
\]
In particular, any optimal solution of \eqref{maxi} that pools at
$\flatq$ below $\cutoff$ coincides almost everywhere with
$\max\{\qmir,\qrel\}$ above $\cutoff$.
\end{enumerate}
\end{proposition}

Part (i) is Lemma~\ref{lmc} in $\csminus$. Part (ii) is the landing
argument above. In part (iii) the competitor $q$ is any implementable
allocation with the same pool: it need not be monotone and need not
stay in $\csplus$. The proof, in Appendix~\ref{app:shape}, is a
weak-duality bound. The constraints \eqref{eq:G_main} are dualized with
the measure $-\Lambda'$ on $(\cutoff,\theta_3]$ plus an atom
$\Lambda(\theta_3)$ at $\theta_3$; by (M) the measure is non-negative,
by (C) the Lagrangian is concave in the decision at every type with
stationary point $\qmir$, and along the mirror every dualized constraint
holds with equality, so the bound is attained at $r$. After $\theta_3$
the relaxed solution is above the mirror, the constraints are slack, and
the contract follows $\qrel$.

\begin{remark}[Scope of the proposition and role of (M)]
\label{rem:mirror_scope}\label{rem:M_role}
The proposition takes the pool as given. It does not say that an
optimum pools on an initial segment below $\qlim$ (an optimum may start
in $\csplus$, Supplemental Example~\ref{ex:case2}), and it does not determine
$(\flatq,\cutoff)$: which pool, and whether to jump at all rather than
follow the continuous candidate of Theorem~\ref{proposicion}, is the
two-variable problem of Proposition~\ref{prop_char}(iii), whose
solution is proved optimal by Corollary~\ref{cor:cert_mirror}.
Condition (M) is not a technical assumption. The constraints
\eqref{eq:G_main} are a state constraint on the accumulated slack and
the mirror is the arc on which it is active; following the arc is
proved optimal when the multiplier measure on the contact set,
$-\Lambda'$, is non-negative; where $\Lambda$ increases that measure
is not non-negative and the proposition does not apply. In the linear-curves family (M) always holds,
since $\Lambda$ is a ratio of affine functions with constant-sign
Wronskian (Supplemental Corollary~\ref{cor:linear_reduction}); in the
guiding family $\Lambda'(t)=(1-\gamma)(b-j)/(\alpha+\gamma t-b)^2<0$, and
(C) is $\Lambda(t)<1$ for $t>\cutoff$, which holds because
$\Lambda(\cutoff)=1$ and $\Lambda$ decreases.
\end{remark}

\end{newresult}

\subsection{Scope}\label{sec:scope}

\begin{sergei}
The three results leave the following open. Theorem~\ref{proposicion}
excludes jumps within its class; it does not show that an optimum
belongs to the class. Proposition~\ref{prop_jump_forced} excludes
continuity under separability; it does not locate the jump.
Proposition~\ref{prop_mirror_shape} determines the continuation of a
pool that jumps; it does not say that the optimum pools, it does not
determine the pool or the cutoff, and without Assumption~M it does
not apply. These questions are treated by the candidate
families and the variational system of
Section~\ref{sec:characterization}, and, where the certificate applies,
global optimality is proved by it. The uniqueness given by the certificate is
up to null sets. In the guiding family the mirror candidate beats the
continuous candidate on the domain of regime III;
Section~\ref{sec:MES_solutions} gives the interior solutions. Regime III is solved on the
sub-domain stated there. Outside it the pooled level returned by the
system can exceed the cutoff, the mirror candidate does not exist, and
we do not characterize the optimum; Section~\ref{sec:exactness} says
what is exact about the certificate in that case.
\end{sergei}

\FloatBarrier

\section{The Certificate}
\label{sec:sufficiency}

A candidate obtained from first-order conditions in a reduced family is
stationary in that family. To prove it optimal we use
a certificate: a non-negative measure on the incentive constraints that
bind at the candidate, with which the objective of every implementable
allocation is bounded above by its value at the candidate. The argument
is the sufficiency theorem of \citet[\S8.4]{Luenberger1969} in the form
of \citet[Appendix~A]{AmadorBagwell2013}, with multipliers constructed
so that the Lagrangian is concave at them, as in
\citet{AmadorWerningAngeletos2006}. Here transfers are present, single
crossing fails, and the binding constraints are global and indexed by
pairs of types, so the multiplier is a measure on pairs whose form
depends on the active set. Its weights are the distortion--rent ratio
$R$ of \eqref{eq:ratio} evaluated along the contract, and every
condition on the measure is a sign or monotonicity condition on $R$ and
on the curvature of the primitives, which can be checked once the two
coordinates of the candidate are known. This section states the
certificate for an arbitrary active set; Section~\ref{sec:exante}
constructs the measure for each binding structure. The certificate does not rely on an existence theorem (for which see
\citealp{Carlier2001}): it identifies the optimum by comparison, and
where it applies existence and uniqueness follow.

\subsection{Active set and weight}
\label{sec:ratio}

Write the active set of an implementable allocation $\qopt$ as
\begin{equation}\label{eq:active_set}
	\mathcal{A}(\qopt):=\bigl\{(\theta_a,\theta_e)\in\Theta^2:\;
	\gif{\theta_a}{\theta_e}{\qopt}=0\bigr\},
\end{equation}
the ordered pairs whose global incentive constraint holds with equality.
A dual object is a finite non-negative Borel measure $\mu$ on $\Theta^2$,
written as $d\mu(\theta_a,\theta_e)=dk_{\theta_a}(\theta_e)\,d\mu_A(\theta_a)$,
where $\mu_A$ is the marginal on anchors and $k_{\theta_a}$ the
conditional on endpoints. Two functionals of $\mu$ enter the
certificate: the mass of pairs that straddle a type,
\begin{equation}\label{eq:straddle}
	m(x):=\mu\bigl\{(\theta_a,\theta_e):\ x \text{ lies between }
	\theta_a\text{ and }\theta_e\bigr\},
\end{equation}
counted with the sign of $\theta_e-\theta_a$, and the density $g$ of
$\mu_A$ where it exists. Along a branch on which a family of constraints
binds, $m$ is $R\dens$ evaluated on the branch, with $R$ the
distortion--rent ratio of \eqref{eq:ratio}; on a pool, $m$ is a ratio
of two integrals. In this role $R$ is a shadow price, the counterpart
for rents nonlinear in the action of the multiplier of
\citet{AmadorBagwell2013} on the region of flexibility.\footnote{There
$v_{q\theta}\equiv1$, $R$ reduces to $-w_\pi$, and their multiplier
$1-w_\pi f$ is $R$ times the density up to a normalization; their
condition (c1) corresponds to the monotonicity conditions on $R\dens$ of
Section~\ref{sec:exante}.} The form of $\mu$ depends on the active set;
Figure~\ref{fig:activesets} shows the four active sets that occur in the
class, treated one by one in Section~\ref{sec:exante}.

\begin{figure}[!tb]
\centering
\begin{tikzpicture}[x=2.6cm,y=2.6cm,font=\scriptsize,>=Stealth]
  \def\colsep{1.22}
  \newcommand{\sq}[1]{
    \draw[black!70] (#1,0) rectangle (#1+1,1);
    \draw[black!30,line width=0.4pt] (#1,0) -- (#1+1,1);
    \node[below] at (#1+0.5,-0.14) {anchor $\theta_a$};
  }
  \newcommand{\cp}[1]{
    \draw[->,black!70] (#1,1.4) -- (#1+1.05,1.4) node[right] {$\theta$};
    \draw[->,black!70] (#1,1.4) -- (#1,2.45) node[above] {$q$};
  }
  \node[rotate=90] at (-0.14,0.5) {endpoint $\theta_e$};
  \node[above] at (0.5,2.62) {$D$: no crossing};
  \cp{0}
  \draw[black!35] (0,1.55) -- (1,2.0);
  \draw[black!55,dashed] (0,1.75) -- (1,2.25);
  \draw[black,line width=1.3pt] (0,1.95) -- (0.35,1.95);
  \draw[black,line width=1.3pt] (0.35,1.95) -- (1,2.4);
  \node[font=\tiny,above] at (0.17,1.96) {pool};
  \draw (0.35,1.4) -- (0.35,1.37) node[below,font=\tiny] {$\theta_1$};
  \sq{0}
  \draw[black,line width=2.2pt] (0,0.985) -- (0.35,0.985);
  \draw (0.35,0) -- (0.35,-0.03) node[below,font=\tiny] {$\theta_1$};
  \node[left,font=\tiny] at (0,0.985) {$\bar\theta$};
  \node[align=center,font=\tiny] at (0.55,0.75) {pool $\times\{\bar\theta\}$};
  \begin{scope}[shift={(\colsep,0)}]
  \node[above] at (0.5,2.62) {$D_J$: horizontal crossing};
  \cp{0}
  \draw[black!35] (0,1.9) -- (1,1.9);
  \draw[black!55,dashed] (0,1.55) -- (1,2.35);
  \draw[black,line width=1.3pt] (0,1.72) -- (0.44,1.72);
  \draw[black,dotted] (0.44,1.72) -- (0.44,2.08);
  \draw[black,line width=1.3pt] (0.44,2.08) -- (0.66,2.08);
  \draw[black,line width=1.3pt] (0.66,2.08) -- (1,2.35);
  \draw (0.44,1.4) -- (0.44,1.37) node[below,font=\tiny] {$\theta_1$};
  \draw (0.66,1.4) -- (0.66,1.37) node[below,font=\tiny] {$\theta_2$};
  \sq{0}
  \fill[black!80] (0,0.44) rectangle (0.44,0.66);
  \draw (0.44,0) -- (0.44,-0.03) node[below,font=\tiny] {$\theta_1$};
  \node[left,font=\tiny] at (0,0.44) {$\theta_1$};
  \node[left,font=\tiny] at (0,0.66) {$\theta_2$};
  \node[align=center,font=\tiny] at (0.6,0.25) {lower flat $\times$\\upper flat};
  \end{scope}
  \begin{scope}[shift={(2*\colsep,0)}]
  \node[above] at (0.5,2.62) {$D_R$: monotone crossing};
  \cp{0}
  \draw[black!35] (0,1.62) -- (1,1.95);
  \draw[black!55,dashed] (0,1.45) -- (1,2.35);
  \draw[black,line width=1.3pt] (0,1.62) -- (0.35,1.62);
  \draw[black,dotted] (0.35,1.62) -- (0.35,1.85);
  \draw[black,line width=1.3pt] (0.35,1.85) -- (0.75,2.12);
  \draw[black,line width=1.3pt] (0.75,2.12) -- (1,2.35);
  \node[font=\tiny,above] at (0.55,1.99) {mirror};
  \draw (0.35,1.4) -- (0.35,1.37) node[below,font=\tiny] {$\theta_1$};
  \draw (0.75,1.4) -- (0.75,1.37) node[below,font=\tiny] {$\theta_3$};
  \sq{0}
  \fill[black!80] (0,0.35) rectangle (0.35,0.75);
  \draw (0.35,0) -- (0.35,-0.03) node[below,font=\tiny] {$\theta_1$};
  \node[left,font=\tiny] at (0,0.35) {$\theta_1$};
  \node[left,font=\tiny] at (0,0.75) {$\theta_3$};
  \node[align=center,font=\tiny] at (0.62,0.2) {pool $\times$\\mirror branch};
  \end{scope}
  \begin{scope}[shift={(3*\colsep,0)}]
  \node[above] at (0.5,2.62) {$D_e$: extreme-type anchor};
  \cp{0}
  \draw[black!35] (0,1.5) -- (1,2.2);
  \draw[black!55,dashed] (0,1.7) -- (1,2.1);
  \draw[black,line width=1.3pt] (0,1.82) -- (0.3,1.82);
  \draw[black,line width=1.3pt] (0.3,1.82) -- (0.62,1.99);
  \draw[black,line width=1.3pt] (0.62,1.99) -- (1,2.1);
  \node[font=\tiny,above] at (0.46,1.92) {$\ell_e$};
  \draw (0.3,1.4) -- (0.3,1.37) node[below,font=\tiny] {$\theta_1$};
  \draw (0.62,1.4) -- (0.62,1.37) node[below,font=\tiny] {$\theta_2$};
  \sq{0}
  \draw[black,line width=2.2pt] (0.3,0.015) -- (0.62,0.015);
  \draw (0.3,0) -- (0.3,-0.03) node[below,font=\tiny] {$\theta_1$};
  \draw (0.62,0) -- (0.62,-0.03) node[below,font=\tiny] {$\theta_2$};
  \node[left,font=\tiny] at (0,0.015) {$\underline\theta$};
  \node[align=center,font=\tiny] at (0.46,0.3) {branch $\times\{\underline\theta\}$};
  \end{scope}
\end{tikzpicture}
\caption{The four binding structures. Top row: the contract (solid)
with the dividing curve $\qlim$ (gray) and the relaxed solution $\qrel$
(dashed). Bottom row: the pairs (anchor $\theta_a$, endpoint
$\theta_e$) whose constraint binds, on which the measure of the
certificate is supported. $D$: every pooled anchor against the top type
$\bar\theta$. $D_J$: every anchor of the lower flat against every type
of the upper flat. $D_R$: every pooled anchor against every type of the
	mirror branch, with a possible endpoint atom at $\theta_3$ only when
	the branch reaches $\bar\theta$ without a tail.
$D_e$: every anchor of the branch against the bottom type
$\underline\theta$.}
\label{fig:activesets}
\end{figure}

\subsection{The theorem}
\label{sec:certificate}

\begin{theorem}[Certificate]
	\label{prop_suff}
	Let $\qopt$ be a measurable implementable allocation, and let $\mu$ be a finite non-negative Borel measure on
	$\Theta^2$ such that:
	\begin{enumerate}[(a)]
		\item \emph{Support.} $\mu\bigl(\Theta^2\setminus\mathcal{A}(\qopt)\bigr)=0$.
		\item \emph{Anchor marginal.} The anchor marginal $\mu_A$ has a
		density $g\ge0$ except possibly for atoms. For each atomic anchor $a$,
		with $\alpha_a:=\mu_A(\{a\})$, the singular functional
		\begin{equation}\label{eq:singular_anchor}
		\mathcal S_a(\xi):=-\alpha_a\int_\Theta\!\int_a^{\theta_e}v_\theta(\xi,x)\,dx\,dk_a(\theta_e)
		\end{equation}
		is maximized at $\qopt(a)$ over $Q$; concavity and stationarity suffice
		for an interior candidate.
		\item \emph{Concavity.} The collected density
		\begin{equation}\label{eq:collected}
			\mathcal{D}(\xi,\theta)=\vs(\xi,\theta)\dens(\theta)
			+m(\theta)\,v_\theta(\xi,\theta)
			-g(\theta)\int_\Theta\!\!\int_{\theta}^{\theta_e}
			v_\theta(\xi,x)\,dx\;dk_{\theta}(\theta_e)
		\end{equation}
		is strictly concave in $\xi$ for a.e.\ $\theta$.
		\item \emph{Pointwise maximality.} $\qopt(\theta)$ maximizes
		$\mathcal{D}(\cdot,\theta)$ over $Q$ for a.e.\ $\theta$.
	\end{enumerate}
	Then $\qopt$ solves \eqref{maxi} among all measurable implementable
	allocations, and it is the unique solution up to a null set. Concavity may fail on a null set of types without affecting
	either conclusion.
\end{theorem}

The proof is in Appendix~\ref{app:sufficiency}. Let
$\mathcal{L}(q):=\jfunc(q)+\int\gif{\theta_a}{\theta_e}{q}\,d\mu$. By
(a), $\mathcal{L}(\qopt)=\jfunc(\qopt)$. Since $\mu\ge0$ and $\Phi\ge0$
for implementable $q$, $\jfunc(q)\le\mathcal{L}(q)$. By Fubini's theorem
$\mathcal{L}(q)=\int\mathcal{D}(q(\theta),\theta)\,d\theta
+\sum_a\mathcal S_a(q(a))$, collecting the continuous terms and then
the atomic anchors. By (b)--(d),
$\mathcal{L}(q)\le\mathcal{L}(\qopt)$, with strict inequality unless
$q=\qopt$ a.e. Atoms in (b) arise on a branch anchored at an extreme
type; when $(\theta_e-a)\,v_{qq\theta}>0$ for $k_a$-a.e.\ $\theta_e$,
$\mathcal S_a$ is strictly concave in $\xi$ and the clause reduces to the
vanishing of its derivative at $\qopt(a)$.
Section~\ref{sec:exante} constructs $\mu$ for each of the four binding
structures.

\begin{sergei}
\begin{remark}[Limits of the certificate]\label{rem:limits}
Theorem~\ref{prop_suff} is a sufficient condition and we prove no
converse: if no measure satisfies (a)--(d), the candidate is not
refuted. Two examples show the limits of the certificate, for different
reasons. In one configuration of the linear-curves family the contract
has a flat followed by an excluded stretch at the lower boundary of $Q$.
There the weight that stationarity determines on the flat is negative
on part of its support, so no non-negative measure of the form used in
the other configurations exists; the obstruction is a sign, and
bounding $Q$ from above does not remove it. An excluded stretch requires
a multiplier on the participation constraint, as in
\citet{AmadorBagwell2022} (Supplemental Remark~\ref{rem:case17}). In
Supplemental Example~\ref{ex:case2} the curves are parallel and the surplus is
asymptotically affine in the decision. The weights are non-negative and
(c) fails: on an extreme-type branch the curvature of the collected
density at the contract is $2\vs_{qq}(1-s_1/\qlim')$, which is zero for
parallel curves, and above the contract it is
$\vs_{qq}(\xi,\theta)-\vs_{qq}(q(\theta),\theta)>0$ (Supplemental
Lemma~\ref{lem:De_curvature}); the obstruction is curvature, and no
truncation of $Q$ and no measure on the binding pairs removes it. We do
not know whether that candidate is optimal; the evidence points to a
duality gap (Supplemental Section~\ref{sec:case2}).
A scalar multiplier on the isoperimetric constraint is not a
certificate: it cancels the linear terms only against uniform shifts of
a flat. At a jump type the concavity coefficient vanishes, the collected
density taking equal values at the two levels of the jump, so optimality
and uniqueness hold almost everywhere without a uniform margin there.
\end{remark}

\end{sergei}

\subsection{Random mechanisms}

\begin{corollary}[Random mechanisms]
	\label{prop_stochastic}
	Let $\qopt$ and $\mu$ satisfy Theorem~\ref{prop_suff} with (c) strict
	for a.e.\ $\theta$. Then no incentive-compatible random mechanism ---
	a measurable family of lotteries $\sigma(\theta)$ on $Q$ with the
	transfer paid in expectation --- earns more than $\qopt$, and any that
	earns as much is degenerate at $\qopt(\theta)$ for a.e.\ $\theta$.
\end{corollary}

The objective and the constraints are linear in $\sigma$, so weak
duality and complementary slackness hold as before, and the collected
Lagrangian is $\int\!\!\int\mathcal{D}(\xi,\theta)\,d\sigma_\theta(\xi)\,d\theta$.
By (d) each inner integral is at most $\mathcal{D}(\qopt(\theta),\theta)$,
with equality only if $\sigma_\theta$ is degenerate at $\qopt(\theta)$,
by strict concavity. Where (c) fails the certificate does not exclude a
gain from lotteries; we do not know whether one exists.

\section{Characterization}\label{sec:characterization}
\label{sec:general_apply}

\new{The geometry of $(\qlim,\qrel)$ determines a configuration. Each
configuration has a reduced family of candidates and a variational
system that locates one. The properties that the geometry imposes
beyond this (continuity, a jump, the landing point of a jump) are the
necessary conditions of Section~\ref{sec:continuity}. Global optimality
is proved by the certificate of Section~\ref{sec:sufficiency}. In this
section we define the four families, state the system, construct the
certificate for each family with its conditions on the primitives, show
where the certificate is exact, and give a test that detects a
candidate outside the four families.}
\subsection{The mirror candidate}
\label{sec:shape_mirror}

\ser{Section~\ref{sec:shape_mirror_main} defines the mirror and shows
that a pool that jumps must continue along it. Here we record the
candidate this defines and the conditions for its implementability.}

\begin{definition}[The mirror candidate]
	\label{thm:mirror}
	\new{Given $(\bar q,\theta_1)$ with $\bar q<\qlim(\theta_1)$ and
	$\qrel(\theta_1)<\qlim(\theta_1)<\qrel(\theta)$ for some $\theta>\theta_1$,
	the \emph{mirror candidate} is}
	\begin{equation}
		q(\theta)=\qmir(\theta)\quad\text{for } \theta\in[\theta_1,\theta_3],
		\qquad
		q(\theta)=\qrel(\theta)\quad\text{for } \theta\in(\theta_3,\oth],
		\label{eq:thm_mirror_shape}
	\end{equation}
	preceded by the pool at $\bar q$, where $\qmir$ is the mirror of $\bar q$
	and $\theta_3$ is the first type at which $\qrel$ reaches it,
	$\qmir(\theta_3)=\qrel(\theta_3)$, with $\theta_3:=\oth$ if no such type
	exists in $\Theta$. When $\qrel$ and $\qmir$ meet more than once the
	branch is $\max\{\qmir,\qrel\}$, which is what \eqref{eq:thm_mirror_shape}
	reduces to under a single crossing. It is implementable, and every constraint anchored in
	the pool binds along the branch: $\gif{\theta_a}{\theta_e}{q}=0$ for
	$\theta_a<\theta_1\le\theta_e\le\theta_3$. The set of mirror candidates
	is $D_R$.
\end{definition}

\new{\new{Implementability of a contract of this shape does not follow from
the local condition. The following lemma gives sufficient conditions
(proof in Section~\ref{sec:proofs_char} of the Supplemental Appendix). For a type
$t$ and a decision $x\ge\qlim(t)$ write ${\qmir}_t^{-}(x)\le\qlim(t)$ for the
decision on the other side of the curve with the same marginal rent,
$v_\theta({\qmir}_t^{-}(x),t)=v_\theta(x,t)$; and for $x\le\qlim(t)$ write
${\qmir}_t(x)$ for its mirror above, as in Definition~\ref{def:mirror_main}.

\begin{lemma}[Implementability of a pool followed by a branch]
	\label{lem:impl_branch}
	Let $\qlim$ be increasing with $\csplus$ above, and let $q$ be constant
	at $\bar q$ on $[\uth,\theta_1)$ and non-decreasing on $[\theta_1,\oth]$
	with $q(\theta)\ge\qlim(\theta)$ there. Then $q$ is implementable if
	\begin{enumerate}[(a)]
		\item $\int_{\theta_1}^{\theta}h(q(x),x)\,dx\ge0$ for every
		$\theta\ge\theta_1$, with $h$ as in \eqref{eq:h_main};
		\item $q(s)\ge{\qmir}_t^{-}(q(t))$ for all $\theta_1\le s<t\le\oth$;
		\item $q(\theta_1^{+})\ge{\qmir}_{\theta_1}(\bar q)$ when
		$\bar q<\qlim(\theta_1)$.
	\end{enumerate}
	Condition (a) is also necessary, and (c) is Proposition~\ref{prop_mirror_shape}(ii).
\end{lemma}

The three conditions correspond to the three kinds of pairs.
Condition~(a) is \eqref{eq:G_main}, the constraints of the pool types
against the branch. Condition~(b) is the constraint of a branch type $t$
against an earlier branch type $s$: since $v_{qq\theta}>0$,
$v_\theta(\cdot,t)$ is convex with minimum on the curve, so
$v_\theta(q(s),t)\le v_\theta(q(t),t)$ exactly when $q(s)$ lies between
the two decisions with the marginal rent of $q(t)$. Condition~(c) is the
constraint of a pool type against the branch in the other direction.
The other pairs hold automatically: within the pool $\Phi=0$, and for a
branch type imitated by a later one both decisions are above the curve,
where $v_\theta(\cdot,t)$ is increasing. The mirror candidate satisfies
the three conditions, (a) and (c) with equality on the branch and (b)
because ${\qmir}_t^{-}(\qmir(t))=\bar q\le q(s)$. The continuous
candidates of Section~\ref{sec:spaces} satisfy (c) vacuously,
but condition (b) cannot hold for them, since it is pointwise while the
constraint they bind holds with equality only in the aggregate;
Lemma~\ref{lem:impl_D} below gives the condition that replaces it.}}

\subsection{Four candidate families}
\label{sec:spaces}

\paragraph{(i) No crossing.} Suppose $\qrel$ does not meet $\qlim$.
\new{If $\qrel\subset\csminus$, the LMC fails along $\qrel$
(Section~\ref{sec:necesarias}); the candidates are then full pooling and
allocations in $\csplus$ above $\qlim$, and the certificate decides
between them. We consider} $\qrel\subset\csplus$, with $\qlim$ and $\qrel$
increasing and $\csplus$ above $\qlim$.

The binding global constraint is that of the top type against the
pooled types. The principal pools $[\uth,\cutoff]$ at $q(\cutoff)$ and
serves $[\cutoff,\oth]$ with a strictly increasing $\qiso$ above $\qrel$.
Within the class of Theorem~\ref{proposicion} the optimum is continuous,
so we search in the family of a pool followed by a strictly increasing
branch. \citet{AVP2022} solve this problem: within this family the
first-order conditions give the candidate, determined by the
isoperimetric condition $\gif{\uth}{\cutoff}{q}=0$.
Theorem~\ref{proposicion} supplies the step their argument leaves open
on one side, and the certificate supplies the step it leaves open on the
other. \new{Continuity is necessary for optimality within the class of
Theorem~\ref{proposicion}, and with the LMC it makes the
flat-then-increasing family the natural one to search. It is a reduced
family: the certificate compares the candidate found in it with every
implementable allocation, and this removes the restriction to the class.
The candidate given by the conditions of \citet{AVP2022} is stationary,
which in this class does not imply optimality.}
Corollary~\ref{cor:cert_nocross} proves it optimal.

\new{Let $D$ be the space of pairs $(q,\cutoff)$ with
$\cutoff\in(\uth,\oth)$ and $q$ continuous on $[\uth,\oth]$, constant
equal to $q(\cutoff)$ on $[\uth,\cutoff]$ and $C^1$ on $[\cutoff,\oth]$.
A kink at $\cutoff$ is allowed; the variations below use the right
derivative there.} The isoperimetric constraint reduces to
\begin{equation}
	\wfunc(q,\cutoff) :=
	\int_{\cutoff}^{\oth}
	\bigl[v_\theta(q(\theta),\theta) - v_\theta(q(\cutoff),\theta)\bigr]\,
	d\theta = 0.
	\label{eq:iso_constraint}
\end{equation}
Assumption~\ref{S1} used $v_\theta>0$ to reduce the virtual surplus and make
IR bind only at $\uth$; since $U'(\theta;q)=v_\theta(\qopt(\theta),\theta)$,
non-negativity of the rent along the optimum is checked \emph{a
	posteriori}.\footnote{With an interior minimum of the rent, IR may bind
	at an interior type and the problem requires reformulation; see
	\citet{AVP2022}.}

\new{Implementability of a candidate in $D$ is a condition on the
constraints anchored at the cutoff and on the branch. For
$s\in[\cutoff,\oth)$ write $h_s(z,x):=v_\theta(z,x)-v_\theta(q(s),x)$
and
\begin{equation}\label{eq:G_s}
G_s(t):=\int_s^t h_s(q(x),x)\,dx,\qquad \cutoff\le s<t\le\oth,
\end{equation}
so that $G_{\cutoff}$ is the function $G_q$ of \eqref{eq:G_main}.

\begin{lemma}[Implementability of the continuous candidate]
\label{lem:impl_D}
Let $\qlim$ be increasing with $\csplus$ above, and let $q$ be constant
at $\bar q\ge\qlim(\cutoff)$ on $[\uth,\cutoff)$, continuous, and
strictly increasing on $[\cutoff,\oth]$ with $q>\qlim$ there. Then $q$
is implementable if and only if $G_s(t)\ge0$ for all
$\cutoff\le s<t\le\oth$. This holds if, for every $s$, the function
$x\mapsto h_s(q(x),x)$ changes sign at most once on $[s,\oth]$, from
positive to negative, and $G_s(\oth)\ge0$.
\end{lemma}

For the candidate of Proposition~\ref{prop_char}(i),
$G_{\cutoff}(\oth)=\wfunc=0$ by construction, so
$h_{\cutoff}(\qiso(x),x)$ takes both signs on the branch; if it changes
sign once, from positive to negative, the branch ends below the mirror
of $\bar q$ and the candidate is implementable, as in the linear-curves
family; it can fail otherwise (Supplemental Section~\ref{sec:H_sweep}). This is why
condition (b) of Lemma~\ref{lem:impl_branch}, which is pointwise, cannot
hold in $D$. In the linear-curves family
$h_s(q(x),x)=\kappa\,[q(x)-q(s)]\,[\tfrac12(q(x)+q(s))-\qlim(x)]$, the
second factor is affine in $x$ and positive at $x=s$, so the single
crossing is automatic, $G_{\cutoff}(\oth)=0$ forces the branch to be
less steep than twice the dividing curve, and implementability reduces
to the polynomial inequality $G_s(\oth)\ge0$ on $s\in[\cutoff,\oth]$;
its strict form for $s>\cutoff$ is the slack condition (S) of
Supplemental Lemma~\ref{lem:competitor}. The proof is in
Section~\ref{sec:proofs_char} of the Supplemental Appendix.}

\paragraph{(ii) Crossing, $\qlim$ horizontal.} Now $\qrel$ meets $\qlim$
at an interior $\theta^*$. It satisfies the LMC on one side of $\theta^*$
and violates it on the other, so the contract is flat where the LMC
fails. The jump at $\theta^*$ gives the principal an additional degree
of freedom, the jump level. We consider $\qlim$ constant, $\qrel$
increasing and starting below $\qlim$, and $\csplus$ above.

The candidate space $D_J$ is the analogue of $D$ with two flat pieces:
triples $(\theta_1,\qL,\qH)$ with $\qopt$ equal to $\qL$ on
$[\uth,\theta_1)$, to $\qH$ on $[\theta_1,\theta_2(\qH))$ and to $\qrel$
above, where $\qrel(\theta_2)=\qH$ and $\qL<\qH$. The constraint binds
between the two flat levels,
\begin{equation}
	\wfunc(\theta_1,\theta_2) :=
	\int_{\theta_1}^{\theta_2(\qH)}
	\bigl[v_\theta(\qH,\theta) - v_\theta(\qL,\theta)\bigr]\,d\theta = 0,
	\label{eq:iso_DJ}
\end{equation}
and the reduced profit $\jfunc(\theta_1,\qL,\qH)$ is the integral of
$\vs$ along the three pieces.

\paragraph{(iii) Crossing, $\qlim$ strictly monotone.} In this case the
incentive constraint is not a one-dimensional integral, since
$v_{q\theta}$ depends on $\theta$, and continuity is not forced, since
Theorem~\ref{proposicion} requires \new{$\qrel$ to lie in one
single-crossing region. The candidate is that of
Definition~\ref{thm:mirror}: if the contract pools at $\bar q$ and jumps
at $\theta_1$, the branch is the mirror $\qmir$ of}
Definition~\ref{def:mirror_main}, which moves with $\theta$, until
$\qrel$ meets it at $\theta_3$, where $\qmir(\theta_3)=\qrel(\theta_3)$,
with $\theta_3:=\oth$ if there is no such type. Hence
\begin{equation}
	D_R = \Bigl\{(\bar q,\theta_1)\in \mathbb{R}^2 : \;
	\qopt = \bar q\,\mathbf{1}_{[\uth,\theta_1)}
	+ \qmir\,\mathbf{1}_{[\theta_1,\theta_3]}
	+ \qrel\,\mathbf{1}_{(\theta_3,\oth]}, \;
	\bar q<\qlim(\theta_1)\Bigr\},
	\label{eq:DR_space}
\end{equation}
with two free parameters. Along the mirror the isoperimetric constraint
holds with equality, so there is no multiplier to determine and the
problem on $D_R$ is an unconstrained maximization in two variables. The
mirror depends on $\bar q$: differentiating
$v_\theta(\qmir(\theta),\theta)=v_\theta(\bar q,\theta)$ gives the
transmission rate
\begin{equation}
	\omega(\theta) := -\frac{\partial\qmir(\theta)}{\partial\bar q}
	= -\frac{v_{q\theta}(\bar q,\theta)}{v_{q\theta}(\qmir(\theta),\theta)}
	> 0 ,
	\label{eq:omega}
\end{equation}
positive because the pool is in $\csminus$ and the mirror in $\csplus$.
Profit is
\begin{equation}
	\jfunc(\bar q,\theta_1) =
	\int_{\uth}^{\theta_1}\vs(\bar q,\theta)\,\dens(\theta)\,d\theta
	+ \int_{\theta_1}^{\theta_3}\vs(\qmir(\theta),\theta)\,\dens(\theta)\,
	d\theta
	+ \int_{\theta_3}^{\oth}\vs(\qrel(\theta),\theta)\,\dens(\theta)\,
	d\theta .
	\label{eq:profit_mirror}
\end{equation}

\new{\paragraph{A fourth family: pool, extreme-type average, relaxed tail.}
The three spaces are reduced families of candidates and not an
exhaustive list. The linear-curves family shows what can lie outside
them.

\begin{definition}[Extreme-type average]
	\label{def:extreme_average}
	For an extreme type $\theta_e\in\{\uth,\oth\}$ and a type $a$, let
	$\ell_e(a)$ be the decision at which $\theta_e$ is indifferent between
	its own contract and that of $a$ along a contract that is constant
	between them:
	\[
	\int_{\theta_e}^{a}v_{q\theta}\bigl(\ell_e(a),t\bigr)\,dt=0 .
	\]
	It is unique by Assumption~\ref{S3}, and when $v_{q\theta}$ is affine in
	$q$ it is the average of $\qlim$ over the types between $\theta_e$ and
	$a$. The family $D_e$ consists of the contracts that pool at
	$\ell_e(\theta_1)$ on the stretch adjacent to $\theta_e$, follow
	$\ell_e$ on $[\theta_1,\theta_2]$, and follow $\qrel$ beyond
	$\theta_2$, with $\qrel(\theta_2)=\ell_e(\theta_2)$, or a terminal flat
	at $\ell_e(\theta_2)$ when $\qrel$ does not reach $\ell_e$.
\end{definition}

The branch is determined by a continuum of constraints binding against
$\theta_e$: if $\gif{a}{\theta_e}{q}=0$ on an interval, differentiating in
$a$ gives $q'(a)\int_{\theta_e}^{a}v_{q\theta}(q(a),t)\,dt=0$, so where
the contract is not constant it lies on $\ell_e$. The reduced profit is a
function of $(\theta_1,\theta_2)$, polynomial in the linear-curves
family.

The certificate for this family is Corollary~\ref{prop_suff_De} of
Section~\ref{sec:exante}, proved in Section~\ref{sec:corners} of the
Supplemental Appendix. Two rows of the family of
Section~\ref{sec:testbed} are solved by it, one anchored at each extreme,
and the allocation of Section~\ref{sec:case2} of the Supplemental
Appendix that beats every pool under relaxed-LMC failure belongs to
$D_e$ anchored at the top type.}

\subsection{The variational system}
\label{sec:char_system}

In each space an optimal candidate is stationary for the reduced
problem, and the multiplier of the isoperimetric constraint is a scalar
$\nu$. The number of conditions depends on the number of pieces of the
contract. The no-crossing case is due to \citet{AVP2022}; the two
crossing cases are new.

\begin{proposition}[Characterization]
	\label{prop_char}
	\label{prop_variacional}
	\label{prop_cross}
	\label{prop_mirror_char}
	Let the candidate be optimal in the space prescribed for its
	configuration, for $\max\jfunc$ subject to $\wfunc=0$. Then:
	\begin{enumerate}[(i)]
		\item \emph{(No crossing; Proposition~1 of \citealp{AVP2022}, for the class of Theorem~\ref{proposicion}.)} For
		$(q,\cutoff)\in D$ there exists $\nu\in\mathbb{R}$ with
		\[
		\vs_q(q(\theta),\theta)\,\dens(\theta)
		+\nu\,v_{q\theta}(q(\theta),\theta)=0
		\quad\text{on }[\cutoff,\oth],
		\]
		\[
		\int_{\uth}^{\cutoff}\vs_q(q(\cutoff),\theta)\,\dens(\theta)\,d\theta
		=\nu\int_{\cutoff}^{\oth}v_{q\theta}(q(\cutoff),\theta)\,d\theta .
		\]
		\item \emph{(Crossing, $\qlim$ horizontal.)} For
		$(\theta_1,\qL,\qH)\in D_J$ there exists $\nu\in\mathbb{R}$ with
		\[
		\bigl[\vs(\qL,\theta_1)-\vs(\qH,\theta_1)\bigr]\dens(\theta_1)
		=\nu\bigl[v_\theta(\qH,\theta_1)-v_\theta(\qL,\theta_1)\bigr],
		\]
		\[
		\int_{\uth}^{\theta_1}\vs_q(\qL,\theta)\,\dens(\theta)\,d\theta
		=\nu\int_{\theta_1}^{\theta_2}v_{q\theta}(\qL,\theta)\,d\theta ,
		\]
		\[
		\int_{\theta_1}^{\theta_2}
		\bigl[\vs_q(\qH,\theta)\,\dens(\theta)+\nu\,v_{q\theta}(\qH,\theta)
		\bigr]d\theta
		+\nu\bigl[v_\theta(\qH,\theta_2)-v_\theta(\qL,\theta_2)\bigr]
		\frac{1}{\qrel'(\theta_2)}=0 ,
		\]
		together with $\wfunc(\theta_1,\theta_2)=0$.
		\item \emph{(Crossing, $\qlim$ monotone.)} For
		$(\bar q,\theta_1)\in D_R$ the virtual surplus is continuous across
		the jump,
		\begin{equation}
			\vs(\bar q,\theta_1)=\vs(\qmir(\theta_1),\theta_1),
			\label{eq:foc_jump}
		\end{equation}
		and the pooled level equalizes marginal virtual surplus between the
		pool and the mirror, weighted by the transmission rate
		\eqref{eq:omega}:
		\begin{equation}
			\int_{\uth}^{\theta_1}\vs_q(\bar q,\theta)\,\dens(\theta)\,d\theta
			=\int_{\theta_1}^{\theta_3}
			\vs_q(\qmir(\theta),\theta)\,\omega(\theta)\,\dens(\theta)\,
			d\theta .
			\label{eq:foc_pool}
		\end{equation}
		No multiplier appears: the mirror equates rents
		at the two levels by construction, so $\wfunc\equiv0$ on $D_R$.
		Here the system determines each candidate but not which one is
		optimal: $D_R$, the continuous contracts of $D$, full pooling, two
		flat pieces, and \new{the extreme-type family $D_e$ are all
		admissible, and the optimum is the one of highest $\jfunc$; where
		the crossing runs from $\csplus$ to $\csminus$ the mirror candidate
		is not defined and the comparison runs among the others.}
	\end{enumerate}
\end{proposition}

\begin{proof}
	Part~(i) is proved in \citet{AVP2022}; Section~\ref{sec:proofs_char}
	of the Supplemental Appendix gives the translation of notation and the
	step that Theorem~\ref{proposicion} supplies. Parts~(ii) and~(iii) are
	proved there as well.
\end{proof}

In~(i) the contract is continuous at the cutoff, so the boundary terms
of the variation cancel. In~(ii) they differ by the profit gap between
the two levels; this gives the jump condition. In~(iii) the mirror moves
with $\bar q$ at the rate \eqref{eq:omega}, and the boundary terms at
$\theta_3$ cancel because $\qmir$ and $\qrel$ meet there (and there is
none when $\theta_3=\oth$, a fixed endpoint). Only in~(iii)
is the shape decided by a comparison of profits. The comparison across
the configurations, including the corner candidates, is in
Section~\ref{sec:corners} of the Supplemental Appendix and
Section~\ref{sec:comparison} of Part~II of the Supplementary Material. In the guiding family the mirror candidate wins on the
sub-domain of regime~III stated in Section~\ref{sec:MES_solutions}; outside
that sub-domain we do not characterize the optimum.

\new{\paragraph{(iv) The fourth family.} A candidate in $D_e$ is
$(\theta_1,\theta_2)$; its reduced profit is
$\jfunc(\theta_1,\theta_2)=\int f(\ell_e(\theta_1),\cdot)$ over the pool
$+\int f(\ell_e(t),t)\,dt$ over the branch $+\int f(\qrel,\cdot)$ over the
tail. An interior candidate solves $\partial_{\theta_1}\jfunc=
\partial_{\theta_2}\jfunc=0$; the boundary terms cancel because the
contract is continuous at both junctions. Implementability has the same
form as in Lemma~\ref{lem:impl_D}, with the anchors on the branch and
the tail and the extreme type as endpoint. We state it for
$\theta_e=\oth$, the case of regime~IV; the case $\theta_e=\uth$ is its
reflection (Section~\ref{sec:decreasing} of the Supplemental Appendix).

\begin{lemma}[Implementability of the extreme-type candidate]
\label{lem:impl_De}
Let $\qlim$ be increasing with $\csplus$ above, and let $q$ be
non-decreasing and continuous, equal to $\qrel$ on $[\uth,\theta_2)$,
to $\ell_{\oth}$ on $[\theta_2,\theta_1]$ and to $\bar q:=\ell_{\oth}(\theta_1)$
on $(\theta_1,\oth]$, with $q>\qlim$ on $[\uth,\theta_1]$. With $h_a$ and
$G_a$ as in \eqref{eq:G_s}, $q$ is implementable if and only if
$G_a(e)\ge0$ for all $a<e$ with $a\le\theta_1$; for $a\in[\theta_2,\theta_1]$,
$G_a(\oth)=0$ by definition of $\ell_{\oth}$. This holds if, for every
$a\le\theta_1$, the function $x\mapsto h_a(q(x),x)$ changes sign at most
once on $[a,\oth]$, from positive to negative, and $G_a(\oth)\ge0$ for
$a<\theta_2$.
\end{lemma}

The pool at the top may lie in $\csminus$; only the tail and the branch
are required to lie above $\qlim$. In the linear-curves family the
factor $\tfrac12(q(x)+q(a))-\qlim(x)$ of $h_a(q(x),x)$ is piecewise
affine and continuous in $x$, with slopes $\tfrac12\qrel'-\qlim'$ on the
tail, $-\tfrac34\qlim'$ on the branch and $-\qlim'$ on the pool, all
negative when the dividing curve is steeper than the relaxed solution,
which is the geometry of this family; the single crossing is then
automatic, the branch anchors are closed by construction, and
implementability reduces to $G_a(\oth)\ge0$ for the anchors of the
tail, a polynomial inequality in $a$. The proof is in
Section~\ref{sec:proofs_char} of the Supplemental Appendix.}

\new{A candidate outside the four families is detected by a Farkas
alternative to the certificate, Proposition~\ref{prop:detection} of
Part~II of the Supplementary Material: when no non-negative measure on the active pairs
makes the candidate stationary, the linear program returns an improving
implementable direction. This is how two configurations of the family,
which had passed every numerical check, were corrected.}

\subsection{The certificate for each structure}
\label{sec:exante}

\ser{There are four binding structures, one for each family of
Section~\ref{sec:spaces}: a pool against the top type ($D$), a pool
against a flat branch ($D_J$), a pool against a moving branch ($D_R$),
and a branch against an extreme type ($D_e$); Figure~\ref{fig:activesets}
shows the four active sets.} In each case the density and the boundary
values of the measure are given by the variational system of
Proposition~\ref{prop_char}, and the conditions under which the measure
proves optimality are conditions on $R$: on a branch the weight is
$R\dens$ along the contract and the condition is its monotonicity; on a
pool the weight is the integrated ratio $F/V_q$ and the condition is a
comparison of rates; the concavity condition compares the curvature of
the surplus with the mass of the rent terms. Throughout, $\qopt$ is the
candidate of the family, $\bar q^*$ its pooled level and $\cutoff^*$
its cutoff.

\paragraph{No crossing.}\label{sec:suff_nocross}
The pool $[\uth,\cutoff^*]$ binds against the top type,
$\gif{\theta_a}{\oth}{\qopt}=0$ for every $\theta_a$ in the pool, and no
constraint binds on the branch. Take $k_{\theta_a}=\delta_{\oth}$ for
every anchor, so that $m\equiv\mu_A(\Theta)$ on the branch. Let
$F(\theta):=\int_{\uth}^{\theta}\vs_q(\bar q^*,t)\dens(t)\,dt$,
$V_q(\bar q^*,\theta):=\int_{\theta}^{\oth}v_{q\theta}(\bar q^*,t)\,dt$, and
\begin{equation}
	A(\theta):=\frac{F(\theta)}{V_q(\bar q^*,\theta)},
	\qquad \lambda:=A' \quad\text{on }[\uth,\cutoff^*].
	\label{eq:flow}
\end{equation}
In this section $A$, with subscripts in the crossing structures, is the
accumulated anchor weight on a pool, and $\Lambda$ is the straddling
mass on a distorted branch, $R\dens$ evaluated along the contract.
Three properties of \eqref{eq:flow} hold identically: $A\ge0$,
$\lambda(\cutoff^*)=0$, and
\begin{equation}
	A(\cutoff^*)=\nu^*,
	\label{eq:mass_nocross}
\end{equation}
where $\nu^*$ is the scalar of the variational system; this is the
transversality condition, and we call \eqref{eq:mass_nocross} the
\emph{mass identity}: the total mass of the measure on the pool equals
the multiplier on the branch.

\begin{assumption}[Mangasarian--Arrow concavity\protect\footnotemark]
	\footnotetext{The concavity condition of the isoperimetric problem
		\citep{Mangasarian1966,SeierstadSydsaeter1987}; each weight assumption
		below is its companion sign condition. Measure-valued multipliers are
		standard in optimal control with state constraints
		\citep{HartlSethiVickson1995}, price participation in
		\citet{Jullien2000} and sweep in \citet{RochetChone1998}.}
	\label{S5}
	The scalar $\nu^*$ of the no-crossing candidate satisfies
	\begin{equation}
		\vs_{qq}(q,\theta)\,\dens(\theta)+\nu^*\,v_{qq\theta}(q,\theta)<0
		\quad\forall\,(q,\theta)\in Q\times\Theta.
		\label{eq:S5}
	\end{equation}
\end{assumption}

\begin{assumption}[Multiplier weight]
	\label{Sflow}
	$V_q(\bar q^*,\uth)<0$, and $A$ of \eqref{eq:flow} is
	non-decreasing on $[\uth,\cutoff^*]$.
\end{assumption}

\begin{corollary}[No crossing]
	\label{cor:cert_nocross}
	Let $\qopt\in D$ be implementable and solve the variational system of
	Proposition~\ref{prop_char}(i). Under Assumptions~\ref{S5}
	and~\ref{Sflow}, the measure with anchor density $\lambda$ on the pool
	and endpoint $\oth$ satisfies (a)--(d) of Theorem~\ref{prop_suff}.
\end{corollary}

Monotonicity of $A$ compares two rates on the pool, the rate at which
the pool accumulates foregone surplus and the rate at which it loses
reference rent, both evaluated at $(\bar q^*,\cutoff^*)$. On the branch
Assumption~\ref{S5} is $-(\dens/v_{q\theta})W<0$ with
$W=\vs_qv_{qq\theta}-\vs_{qq}v_{q\theta}$, that is, Assumption~H
evaluated on the contract (Supplemental Example~\ref{ex:H_fails}).

\paragraph{Horizontal crossing.}\label{sec:suff_cross}
Let $\qlim\equiv c$ and $\qopt\in D_J$, with flats $\qL^*<c<\qH^*$ on
$[\uth,\theta_1^*)$ and $[\theta_1^*,\theta_2^*)$ and then $\qrel$. Every
anchor on the lower flat binds against every type of the upper flat. On
the upper flat the straddling mass is
\begin{equation}
	\Lambda(x):=R(\qH^*,x)\,\dens(x),
	\label{eq:flowH}
\end{equation}
which is positive because $\qH^*$ is above $\qrel$ and $\qlim$ there, and
vanishes at $\theta_2^*$, where the distortion ends. The endpoint kernel
is $k:=-\Lambda'/\Lambda(\theta_1^*)$, a probability density on
$[\theta_1^*,\theta_2^*]$, and the anchors of the lower flat have a
density $\lambda_L$ determined by the pointwise condition at $\qL^*$.
With
\begin{equation}
	W(x):=\int k(\theta_e)\int_{x}^{\theta_e}v_{q\theta}(\qL^*,t)\,dt\,d\theta_e ,
	\qquad
	A_L(x):=\frac{\int_{\uth}^{x}\vs_q(\qL^*,t)\dens(t)\,dt}{W(x)},
	\label{eq:flowL_closed}
\end{equation}
the anchor weight is $\lambda_L=A_L'$, and the two weights coincide at the cutoff, the mass identity of this structure,
\begin{equation}
	A_L(\theta_1^*)=\Lambda(\theta_1^*),
	\label{eq:mass_identity}
\end{equation}
which holds identically.

\begin{assumption}[Multiplier weights, horizontal crossing]
	\label{S6}
	$\Lambda$ of \eqref{eq:flowH} is non-increasing on
	$[\theta_1^*,\theta_2^*]$; $A_L$ of \eqref{eq:flowL_closed} is
	non-decreasing; and $\vs_{qq}(\xi,x)\dens(x)+m(x)v_{qq\theta}(\xi,x)<0$
	for all $\xi$ and a.e.\ $x$, with $m=A_L$ on the lower flat, $\Lambda$ on
	the upper one, $0$ on the tail.
\end{assumption}

\begin{corollary}[Horizontal crossing]
	\label{cor:cert_horizontal}
	Let $\qopt\in D_J$ be implementable and solve the system of
	Proposition~\ref{prop_char}(ii). Under Assumption~\ref{S6}, the measure
	with anchor density $\lambda_L$ on the lower flat and endpoint kernel
	$k$ satisfies (a)--(d).
\end{corollary}

\paragraph{Monotone crossing.}\label{sec:suff_mirror}
Let $\qopt\in D_R$. The pool binds against the whole mirror branch,
$v_\theta(\qmir(\theta),\theta)=v_\theta(\bar q^*,\theta)$ on
$[\theta_1^*,\theta_3]$, and the straddling mass on the branch is the mirror weight
$\Lambda(x)=R(\qmir(x),x)\,\dens(x)$ of \eqref{eq:Lambda_mirror}.
If an undistorted tail $\qrel$ is present, then
$\Lambda(\theta_3)=0$ and the endpoint kernel has no atom. If there is
no tail, then $\theta_3=\bar\theta$ and $\Lambda(\bar\theta)\ge0$, with
equality only if the mirror meets $\qrel$ exactly at the top; otherwise
the endpoint kernel carries the residual atom
\begin{equation}
	dk=\frac{-\Lambda'(\theta_e)}{\Lambda(\theta_1^*)}\,d\theta_e
	+\frac{\Lambda(\theta_3)}{\Lambda(\theta_1^*)}\,\delta_{\theta_3},
	\label{eq:kernel_atom}
\end{equation}
a probability measure on $[\theta_1^*,\theta_3]$ when $\Lambda$ is
non-negative and non-increasing, and positive before the branch ends. The anchor weight $A_R$ is \eqref{eq:flowL_closed} with
$dk$ in place of $k(\theta_e)\,d\theta_e$.

\begin{assumption}[Multiplier weights, monotone crossing]
	\label{S8}
	$\Lambda$ of \eqref{eq:Lambda_mirror} is non-increasing, positive on
	$[\theta_1^*,\theta_3)$ and non-negative on the closed interval; it is
	zero at $\theta_3$ with a tail, while without one $\theta_3=\bar\theta$
	and $\Lambda(\bar\theta)\ge0$, positive unless $\qmir(\oth)=\qrel(\oth)$. Moreover, $W$ is negative on the pool; $A_R$ is
	non-decreasing; and $\vs_{qq}(\xi,x)\dens(x)+m(x)v_{qq\theta}(\xi,x)<0$
	for all $\xi$ and a.e.\ $x$, with $m=A_R$ on the pool, $\Lambda$ on the
	branch, $0$ on the tail.
\end{assumption}

\begin{corollary}[Monotone crossing]
	\label{cor:cert_mirror}
	Let $\qopt\in D_R$ be implementable and solve the system of
	Proposition~\ref{prop_char}(iii). Under Assumption~\ref{S8}, the
	measure with anchor density $\lambda_R=A_R'$ on the pool and endpoint
	kernel \eqref{eq:kernel_atom} satisfies (a)--(d). Its possible residual
	atom is an endpoint atom, already included in $m=\Lambda$; since
	$\mu_A=\lambda_R\,d\theta$, no singular term \eqref{eq:singular_anchor} arises.
\end{corollary}

The corollary has no second-order condition within $D_R$. A candidate
proved optimal among all implementable allocations is a maximum of
$\jfunc$ within $D_R$, so negative definiteness of the Hessian of
$\jfunc$ in $(\bar q,\theta_1)$ at the candidate follows from
Assumption~\ref{S8}. In the linear-curves family the two conditions have
the same sign, a relaxed solution steeper than the dividing curve
(Supplemental Lemma~\ref{lem:anchor_clause}).

\paragraph{Extreme-type anchor.}
\label{sec:suff_De}
Let $\qopt\in D_e$ be anchored at $\theta_e=\uth$. Every type of the
branch, and every pooled type, binds against the bottom type. Let
$\bar B(t)$ be the mass of anchors above $t$:
\begin{equation}
	\bar B=\frac{F}{V}\ \text{on the pool},\qquad
	\bar B(t)=\frac{\vs_q(\ell_e(t),t)\dens(t)}{v_{q\theta}(\ell_e(t),t)}\
	\text{on the branch},\qquad \bar B=0\ \text{on the tail},
	\label{eq:flow_De}
\end{equation}
with $F,V$ as in \eqref{eq:flow} but integrated from $\uth$. On the
branch the cancellation of the linear terms is pointwise, because
$\int_{\uth}^{t}v_{q\theta}(\ell_e(t),s)\,ds=0$ by definition of $\ell_e$.
If the branch reaches the opposite extreme $\oth$ before $\qrel$ does,
then $\bar B(\oth)>0$ and the measure has an atom at the anchor $\oth$.
Condition (b) also requires singular maximality: $\ell_e(\oth)$ gives
stationarity, and $\mathcal S_{\oth}$ is strictly concave in $\xi$ when
$v_{qq\theta}<0$, so the atom is maximized at $\ell_e(\oth)$. For
$\theta_e=\oth$ exchange endpoints and signs.

\begin{assumption}[Multiplier weight, extreme-type anchor]
	\label{S_De}
	$\bar B\ge0$ is non-increasing and continuous at $\theta_1$ and
	$\theta_2$, so that $\beta=-\bar B'$ is a non-negative density on the
	branch; and $\vs_{qq}(\xi,t)\dens(t)-\bar B(t)\,v_{qq\theta}(\xi,t)
	+\beta(t)\int_{\uth}^{t}v_{qq\theta}(\xi,x)\,dx<0$ for all $\xi$ and
	a.e.\ $t$. At any atomic extreme anchor $a$, $\mathcal S_a$ of
	\eqref{eq:singular_anchor} is maximized at $\qopt(a)$ (automatic
	when $(\theta_e-a)v_{qq\theta}>0$, by stationarity of $\ell_e$).
	For $\theta_e=\oth$, exchange the roles of the endpoints.
\end{assumption}

\begin{newresult}
On an extreme-type branch with constant curvatures and uniform types the
concavity condition has a closed form. The curvature of the collected
density at the contract is $2\vs_{qq}\,(1-s_1/\qlim')$, where $s_1$ is
the slope of the relaxed curve at the contract, and away from the
contract the term $\vs_{qq}(\xi,t)-\vs_{qq}(q(t),t)$ is added
(Supplemental Lemma~\ref{lem:De_curvature}). The factor $1-s_1/\qlim'$
is the reason regime IV of the guiding family requires $\gamma>1$, and
its vanishing for parallel curves is the reason Supplemental Example~\ref{ex:case2}
is not covered.
\end{newresult}

\begin{corollary}[Extreme-type anchor]
	\label{prop_suff_De}
	Let $\qopt\in D_e$ be implementable and solve the system of
	the reduced profit of Section~\ref{sec:spaces}. Under Assumption~\ref{S_De}, the
	measure with anchor density $\beta$ on the branch and endpoint
	$\theta_e$ satisfies (a)--(d).
\end{corollary}

The value of the weight at the end opposite to the pool determines the
distortion at the top. The weight vanishes where the tail $\qrel$
begins, at $\oth$ in the no-crossing case and at $\theta_2^*$ across a
flat curve, and it is the atom of \eqref{eq:kernel_atom} when the mirror
reaches $\bar\theta$ before meeting $\qrel$. The classical absence of
distortion at the top is the case in which the weight vanishes there.
At the jump type the scalar $\nu^*$ is the mean of the weight,
\eqref{eq:mass_nocross} and \eqref{eq:mass_identity}. In the
linear-curves family the jump is at the crossing and
$\Lambda(\theta_1^*)=-b_2/\kappa$ for every level of the jump
(Supplemental Remark~\ref{rem:jump_affine}); outside the family the
mass identity is the equation that locates the jump.

\subsection{Where the certificate is exact}
\label{sec:exactness}

\begin{newresult}
Theorem~\ref{prop_suff} is a sufficient condition. As in
\citet[\S3.4--3.5]{AmadorBagwell2013}, we ask which of its conditions
are necessary and on which family necessary and sufficient conditions
coincide. When the collected density is convex instead of concave on
the branch, the candidate minimizes the Lagrangian among the
allocations with the same pool, and every one of them earns more.

\begin{proposition}[Necessity of the curvature clause]
\label{prop:necessity_curvature}
Let the no-crossing configuration hold and let $\qopt\in D$ be the
candidate of Proposition~\ref{prop_char}(i), with pool $\bar q^*$ on
$[\uth,\cutoff^*]$, branch $\qiso$ on $(\cutoff^*,\oth]$, scalar $\nu^*$,
and $\wfunc(\qopt)=0$. If
$\vs_{qq}(\xi,\theta)\dens(\theta)+\nu^*v_{qq\theta}(\xi,\theta)\ge0$ for
all $\xi\in Q$ and almost every $\theta\in(\cutoff^*,\oth]$, then every
implementable allocation that coincides with $\qopt$ on the pool and
satisfies $\wfunc=0$ earns at least $\jfunc(\qopt)$, and strictly more
where the inequality is strict on a set of positive measure and the
allocation differs there. In particular $\qopt$ is not optimal as soon as
one such allocation exists.
\end{proposition}

The proof, in Section~\ref{sec:necessity_curvature} of the Supplemental
Appendix, is the certificate read in reverse: with the collected density
convex on the branch, $\qiso$ is its minimizer, and the Lagrangian bounds
profit from below instead of from above. The concavity condition is
therefore necessary. The other conditions are different: on the family
with constant curvatures they follow from the first-order conditions.

\begin{corollary}[Exactness on the family with constant curvatures]
\label{cor:exact_nocross}
In the linear-curves family of the Supplemental Appendix
($\vs_{qq}\equiv b_2<0$, $v_{qq\theta}\equiv\kappa>0$, affine $\qrel$ and
$\qlim$, uniform types), let the candidate of
Proposition~\ref{prop_char}(i) be implementable, with
$V_q(\bar q^*,\uth)<0$ and a non-decreasing branch. Then the monotonicity clause of
Assumption~\ref{Sflow} is a consequence of the variational system, and
the candidate is certified by Corollary~\ref{cor:cert_nocross} if
$b_2+\nu^*\kappa<0$. If $b_2+\nu^*\kappa>0$, the branch is strictly
increasing, and no constraint anchored in the branch binds, the
candidate is not optimal: Supplemental Lemma~\ref{lem:competitor}
supplies the competitor that Proposition~\ref{prop:necessity_curvature}
requires. In the horizontal and mirror structures the anchor weights
are non-decreasing at every solution of the system of
Proposition~\ref{prop_char}(ii) and (iii), with or without the
undistorted tail, and the concavity clause holds almost everywhere by
construction; a stationary candidate of $D_J$ is therefore certified,
and a stationary candidate of $D_R$ is certified under the single sign
$s_1>s_0$, a relaxed solution steeper than the dividing curve, which is
also the sign of the Hessian of the two-variable problem
(Supplemental Lemmas~\ref{lem:pool_clause} and~\ref{lem:anchor_clause}).
\end{corollary}

With constant curvatures the function whose sign gives the monotonicity
of the weight is a quadratic polynomial on the pool. It vanishes at the
cutoff by transversality and its sign there is the sign of the slope of
the branch. This is the property that makes the family of
\citet[Proposition~6]{AmadorBagwell2013} special; Supplemental Example~\ref{ex:H_fails}
shows that outside the family the two conditions can differ. In the
crossing structures there is no condition left to be necessary: the
certificate holds at every stationary candidate of the family, and the
optimum among families is selected by comparing profits.
\end{newresult}

\subsection{Reduction to sign conditions}
\label{sec:cert_reduction}

For a candidate produced by the system above, most of the conditions of
Theorem~\ref{prop_suff} hold automatically.

\begin{lemma}[Certificate reduction for variational candidates]
	\label{lem:cert_reduction}
	Let $\qopt$ satisfy, in its regime, Proposition~\ref{prop_char}, with
	the binding structure of (a) and the anchor-marginal decomposition of
	(b) in Theorem~\ref{prop_suff}. Then the mass
	identities \eqref{eq:mass_nocross} and \eqref{eq:mass_identity},
	the identity $A_R(\theta_1^*)=\Lambda(\theta_1^*)$, and the
	boundary values $\lambda(\theta_1^*)=0$, $\Lambda(\theta_2^*)=0$,
	and $\Lambda(\theta_3)=0$ when $\theta_3<\oth$, hold
	identically. Condition (a), the non-atomic part of (b), and stationarity
	in (d) are automatic; atomic maximality in (b) must still be checked. The certificate
	reduces to the sign, monotonicity, and concavity requirements on the
	weights it constructs. None of this uses linearity.
\end{lemma}

\begin{proof}
	See Section~\ref{sec:proofs_char} of the Supplemental Appendix.
\end{proof}

In the linear-curves family the
remaining conditions reduce to finitely many root conditions
(Supplemental Corollary~\ref{cor:linear_reduction}), which makes most
checks exact and isolates the five partially numerical validations. The two
applications of Section~\ref{sec:applications} have nonlinear
primitives; there only the lemma applies, and each sign and
monotonicity condition is verified directly.

Decreasing configurations reduce to this one through the reflection
$\rho(\theta)=\uth+\oth-\theta$, which leaves the mirror invariant and
reverses the signs of $v_{q\theta}$ and of the multiplier
(Section~\ref{sec:decreasing} of the Supplemental Appendix).

The lemma and the corollary prove the four contracts of
Section~\ref{sec:MES_solutions} optimal. In each, the variational conditions hold,
so the mass identities and the boundary values hold by the lemma; by the
corollary the jump is at the crossing, the concavity condition reduces to
a finite-root condition, and the branch weight is monotone by
construction. Regime~III also satisfies Assumption~\ref{S8} on the
sub-domain of Section~\ref{sec:MES_solutions}. The four interior solutions of
\eqref{eq:MES_v} are therefore optimal among all implementable
allocations on the stated domains. The rest of the class is treated in
Section~\ref{sec:testbed}.

\section{Applications}\label{sec:applications}

We give two applications with nonlinear primitives and no closed-form
solution. The first is the procurement problem of the introduction. The
closed forms of the guiding family of Section~\ref{sec:MES_env} end the
section.

\subsection{Procurement}
\label{sec:outside}

Section~6 of \citet{schottmuller2015} illustrates his algorithm on an amended
version of his running example: the principal's payoff is $u(q,\theta)-t$ and
the agent's is $t-c(q,\theta)$, with
\[
u(q,\theta)=\frac{21\theta q}{10},
\qquad
c(q,\theta)=\theta q+\frac{q^{2}}{2\theta}-\frac{\theta}{10},
\]
and $\theta$ uniform on $\Theta=[1/4,1/2]$. In the procurement
interpretation $\theta$ is the supplier's scale. The linear term $\theta q$
is overhead, heavier for larger suppliers, and the convex term
$q^{2}/(2\theta)$ is congestion, which larger suppliers absorb more
easily. Since $c_{q\theta}=1-q/\theta^{2}$, overhead dominates below
$q=\theta^{2}$ and congestion above it, and the ranking of suppliers by
marginal cost is reversed at the dividing curve. The buyer values output
from larger suppliers more, $u_{q\theta}>0$.

On this example \citet{AVP2022} exhibit an implementable $q^{a}$ that
earns $0.0641923$, against $0.0612817$ for the decision reported in
\citet{schottmuller2015}. The question of what the optimum is was left
open, since $q^{a}$ was constructed to show that the reported decision
can be improved on, not as a solution. We solve it here with the system of
Proposition~\ref{prop_char}(i); $q^{a}$ gives the shape.

Write $v=-c$. Then
$v_{q\theta}=q/\theta^{2}-1$ vanishes on the dividing curve
$\qlim(\theta)=\theta^{2}$, and $v_{qq\theta}=1/\theta^{2}>0$ places
$\csplus$ above it. With $\theta$ uniform the virtual surplus is
$\vs=v+u-(\tfrac12-\theta)v_{\theta}$, and the relaxed solution
$\qrel(\theta)=\theta^{2}(5+\theta)/5$ is above $\qlim$ by $\theta^{3}/5>0$
on $\Theta$: the curves do not meet, and this is a no-crossing
configuration. Within the class of Theorem~\ref{proposicion} the
optimum is continuous, so the candidate is sought in $D$, a pool
followed by the branch; the certificate below compares it with every
implementable allocation, so the restriction to the class is removed.
The first best $q^{\mathrm{FB}}(\theta)=11\theta^{2}/10$ is the
benchmark used in \citet{schottmuller2015}; it lies above $\qrel$, and
the two coincide only at $\overline\theta$. Every pooled type binds
against the top type, so the multiplier is a measure.

The system is the mass identity together
with $\Phi(\overline\theta,\theta_{1})=0$, the isoperimetric condition
against the top type. Its solution is
\begin{equation*}
\theta_{1}=0.311834,
\qquad
\nu^{*}=1.501497,
\qquad
\bar q^{*}=0.121571,
\end{equation*}
rounded to six digits.\footnote{\label{fn:procurement_poly}All quantities are rational in
$\theta$, since $\vs_q=(\theta^3+5\theta^2-5q)/(10\theta^2)$ and
$v_{q\theta}=q/\theta^2-1$, so the system is polynomial: eliminating
$\nu$ by a resultant, $\theta_1$ is the root in $(\tfrac14,\tfrac12)$ of
an integer polynomial of degree ten, and $\nu^*$ and $\bar q^*$ follow
from it. The candidate is therefore algebraic, and we locate $\theta_1$
in a rational interval of width below $10^{-40}$.} These are the parameters
of $q^{a}$: the candidate of \citet{AVP2022} is the one selected by the
first-order conditions. Figure~\ref{fig:procurement} shows the contract. Note that the pool is
strictly inside $Q$, so it is not produced by the bound $q\ge0$, and
that the branch is above $q^{\mathrm{FB}}$ on the whole interval, so
every type is distorted upward. In the procurement interpretation,
every supplier produces beyond its efficient scale.

\begin{figure}[!tb]
\centering
\begin{minipage}[b]{0.48\linewidth}\centering
\begin{tikzpicture}[x=24cm,y=11cm,>=Stealth,font=\footnotesize]
  \begin{scope}[shift={(-0.25,0)}]
    \fill[black!7] plot[domain=0.25:0.5,samples=40] (\x,{\x*\x}) -- (0.5,0.40) -- (0.25,0.40) -- cycle;
    \node[black!60] at (0.275,0.30) {$CS_+$};
    \draw[black!35,line width=0.8pt] plot[domain=0.25:0.5,samples=40] (\x,{\x*\x});
    \node[black!60,anchor=north east] at (0.5,0.245) {$q_0$};
    \draw[black!55,dashed,line width=0.8pt] plot[domain=0.25:0.5,samples=40] (\x,{\x*\x*(5+\x)/5});
    \draw[black!55,dotted,line width=0.8pt] plot[domain=0.25:0.5,samples=40] (\x,{1.1*\x*\x});
    \node[black!60,anchor=west] at (0.5,0.30) {$q^{\mathrm{FB}}$};
    \node[black!60,anchor=west] at (0.5,0.265) {$q_1$};
    \draw[black,line width=1.5pt] (0.25,0.121571) -- (0.311834,0.121571);
    \draw[black,line width=1.5pt] plot[domain=0.311834:0.5,samples=40] (\x,{\x*\x*(0.4*\x+0.498503)/0.498503});
    \node[anchor=south west] at (0.255,0.121571) {$\bar q^*$};
    \node[anchor=west] at (0.455,0.335) {$q^{a}$};
    \draw[->] (0.25,0) -- (0.51,0) node[right] {$\theta$};
    \draw[->] (0.25,0) -- (0.25,0.40) node[above] {$q$};
    \draw (0.25,0) -- (0.25,-0.01) node[below] {$\tfrac14$};
    \draw (0.5,0) -- (0.5,-0.01) node[below] {$\tfrac12$};
    \draw (0.311834,0) -- (0.311834,-0.01) node[below] {$\theta_1^*$};
  \end{scope}
\end{tikzpicture}\\[2pt]
(a) Procurement
\end{minipage}\hfill
\begin{minipage}[b]{0.48\linewidth}\centering
\begin{tikzpicture}[x=6.6cm,y=4.4cm,>=Stealth,font=\footnotesize]
  \begin{scope}[shift={(-0.1,0)}]
    \fill[black!7] (0.1,0.5) -- (1,0.5) -- (1,1.05) -- (0.1,1.05) -- cycle;
    \node[black!60] at (0.2,0.93) {$CS_+$};
    \draw[black!35,line width=0.8pt] (0.1,0.5) -- (1,0.5) node[right,black!60] {$q_0$};
    \draw[black!55,dashed,line width=0.8pt] plot[domain=0.1:1,samples=60] (\x,{\x^(0.75)});
    \node[black!60,anchor=north west] at (0.12,0.17) {$q_1$};
    \draw[black,line width=1.4pt] (0.1,0.375908) -- (0.396850,0.375908);
    \draw[black,dotted,line width=0.8pt] (0.396850,0.375908) -- (0.396850,0.624092);
    \draw[black,line width=1.4pt] (0.396850,0.624092) -- (0.533333,0.624092);
    \draw[black,line width=1.4pt] plot[domain=0.533333:1,samples=40] (\x,{\x^(0.75)});
    \node[anchor=south,font=\footnotesize] at (0.25,0.376) {$q_L^*$};
    \node[anchor=south,font=\footnotesize] at (0.465,0.624) {$q_H^*$};
    \draw[->] (0.1,0) -- (1.04,0) node[right] {$\theta$};
    \draw[->] (0.1,0) -- (0.1,1.08) node[above] {$q$};
    \draw (0.1,0) -- (0.1,-0.02) node[below] {$0.1$};
    \draw (1,0) -- (1,-0.02) node[below] {$1$};
    \draw (0.396850,0) -- (0.396850,-0.02) node[below] {$\theta^*$};
    \draw (0.533333,0) -- (0.533333,-0.02) node[below] {$\theta_2^*$};
  \end{scope}
\end{tikzpicture}\\[2pt]
(b) Nonlinear pricing
\end{minipage}
\caption{The two applications; contracts are solid, the relaxed solution
$\qrel$ dashed and the dividing curve $\qlim$ gray, with $\csplus$ shaded.
(a) Procurement (Section~\ref{sec:outside}): $q^{a}$ pools on
$[\tfrac14,\theta_1^*)$ and follows the branch $\qiso$ above, distorted
upward everywhere. (b) Nonlinear pricing (Section~\ref{sec:app_nlp}): the
optimal contract jumps at $\theta^*\approx0.397$, where $\qrel$ crosses
the flat dividing curve, with a flat on each side of the jump, and
follows $\qrel$ above $\theta_2^*$.}
\label{fig:procurement}\label{fig:NLP}
\end{figure}

\paragraph{Certificate.} The weight $A=F/V_{q}$ of \eqref{eq:flow} has a
closed form, and the certificate of Corollary~\ref{cor:cert_nocross} is
verified at the candidate in two steps. Three conditions are exact
identities, evaluated at the algebraic parameters: $V_q(\bar q^*,\underline\theta)
=-0.0069<0$; the mass identity \eqref{eq:mass_nocross} gives
$A(\theta_1)=\nu^*$, a consequence of the necessary conditions
(footnote~\ref{fn:procurement_poly}); and Assumption~H holds without computation, since
$\partial_qR=1/(10\theta\,v_{q\theta}^2)>0$ identically. Two conditions
hold over the whole pool $[\tfrac14,\theta_1)$ rather than at a point,
and are verified without a grid. The density is strictly concave in
the decision: its second derivative is $(m(\theta)-2)/\theta^2$, with
$m=A\le\nu^*$ on the pool and $m=\nu^*$ on the branch, so it is at most
$(\nu^*-2)/\overline\theta^{\,2}\approx-1.994$ on $Q\times\Theta$. And
$A$ is increasing, $\lambda\ge0$: $\lambda V_q^2=\vs_q\dens V_q+Fv_{q\theta}$
is a rational function of $\theta$ that vanishes at $\theta_1$, and
interval arithmetic, with the parameters enclosed as above, shows it is
positive on $[\tfrac14,\theta_1)$ (script \texttt{procurement\_exact.py}
in the deposit). Given
concavity, $q^{a}$ solving the Euler equation is automatically the
pointwise maximizer of the density, on the pool and on the branch; this
is not a separate check. \new{The sign condition
of Assumption~\ref{S1} fails here: $v_\theta(q^{a}(\theta),\theta)$
changes sign at $\theta=0.485591$, so the reduced problem is a
relaxation of the original one. The rent of $q^{a}$ increases and then
decreases by $3.5\times10^{-5}$, and remains positive above
$\underline\theta$, so $q^{a}$ is feasible for the original problem;
optimality in the relaxation implies optimality in the original problem,
and participation binds only at $\underline\theta$.} By
Theorem~\ref{prop_suff} no implementable allocation earns more than
$q^{a}$, and by Corollary~\ref{prop_stochastic} the same holds for random
mechanisms.

\subsection{Nonlinear pricing}
\label{sec:app_nlp}
\label{sec:schottmuller}

A monopolist sells quality $q$ to consumers with expertise
$\theta\in[\uth,\oth]=[0.1,\,1]$, uniformly distributed. Expertise enters
in two ways. The consumer's gross utility is
\begin{equation*}
	v(q,\theta) =
	\theta\Bigl(\frac{q^2}{2} - \frac{q}{2}\Bigr) + \theta + q,
\end{equation*}
so that expertise raises the marginal value of quality at high levels
and lowers it at low levels. Expert consumers are more demanding at high
quality and cheaper to serve at the margin, with decreasing returns:
\begin{equation*}
	C(q,\theta) =
	\theta(q^2-q) - \theta^{3/4}\,q + \frac{3}{2}q.
\end{equation*}
The virtual surplus is
$\vs(q,\theta) = \tfrac{q}{2}\bigl(2\theta^{3/4} - q\bigr)$, strictly
concave in $q$ as in Assumption~\ref{S2}. Since
$v_\theta(q,\theta) = q^2/2 - q/2 + 1 \geq 7/8 > 0$, the rent is strictly
increasing along any allocation and \eqref{eq:IR} binds only at $\uth$.

Single crossing fails: $v_{q\theta}(q,\theta) = q - \tfrac{1}{2}$ changes
sign at $q = \tfrac{1}{2}$, so $\qlim \equiv \tfrac{1}{2}$. The relaxed
solution $\qrel(\theta) = \theta^{3/4}$ is nonlinear and increasing; it
starts below $\qlim$ and crosses it at
$\theta^{*} = \qlim^{4/3} = (1/2)^{4/3} \approx 0.397$. This is the
horizontal crossing case, and a jump is unavoidable.

Proposition~\ref{prop_jump_forced}
excludes continuity; \new{its jump condition holds here because $\vs$
is quadratic in $q$ and $\varphi(q)=q^2/2-q/2$ is symmetric about
$\tfrac12$. Lemma~\ref{lmc}, with ironing on each side of $c$, leaves
two} flat pieces separated by a jump and then $\qrel$: the candidate is
$\qL$ on $[\uth,\theta_1)$, $\qH$ on $[\theta_1,\theta_2)$ and $\qrel$
above, with three free parameters. Every anchor of the lower flat binds
against every type of the upper flat, so the multiplier is a measure; no
type is excluded, so it has no atom. Proposition~\ref{prop_cross} gives the system.

Since $v_\theta$ does not depend on
$\theta$, the isoperimetric constraint with equality reduces to
\begin{equation}
	\qL + \qH = 1,
	\label{eq:iso_nlp}
\end{equation}
so $v_\theta$ is equal at the two levels, the jump condition is
$\vs(\qL,\theta_1)=\vs(\qH,\theta_1)$, and the jump is at the crossing,
$\theta_1^{*}=\theta^{*}$. With $\qL=1-\qH$ and $\theta_2(\qH)=\qH^{4/3}$,
profit is a function of $\qH$ alone, strictly concave on the admissible
range, and its first-order condition,
$\tfrac{3}{7}\,\qH^{7/3}-\uth\,\qH=\tfrac{3}{7}\,\theta^{*}-\uth+\tfrac{4}{7}\,\uth^{7/4}$,
is a polynomial of degree seven in $\qH^{1/3}$, solved numerically:
$\qL^{*}=0.3759078$, $\qH^{*}=0.6240922$, $\theta_2^{*}=0.5333328$, and
\begin{equation}
	\qopt(\theta) = \begin{cases}
		\qL^{*} & \theta\in[\uth,\;\theta^{*}),\\[4pt]
		\qH^{*} & \theta\in [\theta^{*},\;\theta_2^{*}),\\[4pt]
		\theta^{3/4} & \theta\in[\theta_2^{*},\;\oth].
	\end{cases}
	\label{eq:qopt_nlp}
\end{equation}
Figure~\ref{fig:NLP}(b) shows the contract.

\medskip

\paragraph{Certificate.} Since $\qrel$ is nonlinear, the reduction of
the linear-curves family does not apply, and Assumption~\ref{S6} is
checked directly on \eqref{eq:flowH}, here
$\Lambda(x)=(\qH^{*}-x^{3/4})/(\qH^{*}-\tfrac12)$. The endpoint values
are exact: $\Lambda(\theta_1^{*})=1$ since $\theta^{*3/4}=\qlim=\tfrac12$,
and $\Lambda(\theta_2^{*})=0$ since $\theta_2^{*3/4}=\qH^{*}$. The
constraint stops binding where $\qrel$ reaches the upper flat, the
kernel has no atom, and the top types are not distorted. The rest of
the certificate is analytic. Since $v_{q\theta}=q-\tfrac12$ does not
depend on the type, $W(x)=(\qL^*-\tfrac12)(\bar e-x)$, where
$\bar e=\tfrac37(\qH^{*7/3}-2^{-7/3})/(\qH^*-\tfrac12)\in(\theta_1^*,\theta_2^*)$
is the mean of the kernel, and $\lambda_L\ge0$ is equivalent to
$H(x):=(x^{3/4}-\qL^*)(\bar e-x)+\int_{\uth}^x(t^{3/4}-\qL^*)\,dt\le0$
on the lower flat. Now $H'(x)=\tfrac34x^{-1/4}(\bar e-x)>0$ there, and
$H(\theta_1^*)=0$ is the mass identity \eqref{eq:mass_identity},
$A_L(\theta_1^*)=\Lambda(\theta_1^*)=1$; hence $H<0$ and $\lambda_L>0$
on $[\uth,\theta_1^*)$. Concavity follows: the second derivative of the
density is $-1+A_L\le0$ on the lower flat and $-1+\Lambda\le0$ on the
upper one, with equality only at the jump type, where the density is
affine on $[\qL^{*},\qH^{*}]$ (Supplemental Remark~\ref{rem:jump_affine}).
Only the value of $\qH^*$ is numerical, and none of these conditions
depends on it beyond $\qL^*<\tfrac12<\qH^*$. By
Corollary~\ref{cor:cert_horizontal}, \eqref{eq:qopt_nlp} is optimal
among all implementable allocations.

\subsection{Closed forms for the guiding family}\label{sec:MES_solutions}\label{sec:magic}

\new{Within \eqref{eq:MES_v} the systems of Proposition~\ref{prop_char}
solve in closed form.} The four interior solutions, on the domains
stated in Section~\ref{sec:MES_env}, are the following; their
optimality among all implementable allocations is proved in
Section~\ref{sec:exante}.
\begin{enumerate}[(I)]
\item \emph{No crossing}, $(\alpha,\gamma)=(-a,1)$, $a<\tfrac13$: pool
$\bar q^*=\frac{(3a-2)(3a-1)}{4-9a}$ on $[0,\cutoff^*)$ with
$\cutoff^*=\frac{3a-1}{9a-4}$, then the branch $\theta+b^*$ with
$b^*=\frac{(3a-1)^2}{4-9a}$; the top type is distorted by $b^*$.
\item \emph{Horizontal crossing}, $(\alpha,\gamma)=(a,0)$,
$a<\tfrac1{\sqrt2}$: $\qL^*=(2-\sqrt2)a$ on $[0,a)$, $\qH^*=\sqrt2\,a$ on
$[a,\sqrt2a)$, then $\qrel=\theta$; no distortion at the top.
\item \emph{Monotone crossing, $\qrel$ steeper}, $\gamma<\tfrac12$ and
$j:=\frac{\alpha}{1-\gamma}<\sqrt{\frac{1-2\gamma}{2(1-\gamma)}}$: pool
$b^*=j\bigl[2(1-\gamma)-\sqrt{2(1-\gamma)(1-2\gamma)}\bigr]$ on $[0,j)$, the
mirror $2\qlim-b^*$ on $[j,r)$ with $r=j\sqrt{\frac{2(1-\gamma)}{1-2\gamma}}$,
then $\qrel$; no distortion at the top.
\item \new{\emph{Monotone crossing, $\qlim$ steeper}, $\gamma>1$,
$b:=\frac{-\alpha}{\gamma-1}\in(\tfrac12,1)$, $\alpha+\tfrac\gamma2\ge0$:
the branch $\ell^+(\theta)=\alpha+\gamma\frac{1+\theta}{2}$ on
$[0,2b-1)$, then a pool at $b$; the crossing runs from $\csplus$ to
$\csminus$, the mirror of a pooled level is not defined, and the binding
family is the top type against the branch, with weight
$2(\ell^+-\theta)/(\gamma(1-\theta))$ on the branch, $1/\gamma$ on the
pool, an atom $1+2\alpha/\gamma$ at $\uth$ and gap
$(\gamma-1)/(2\gamma)$. The instance of Figure~\ref{fig:panel_triple},
$(\alpha,\gamma)=(-\tfrac25,\tfrac32)$, has $\cutoff^*=\tfrac35$,
$\bar q^*=\tfrac45$ and profit $563/4000$, against $1/8$ for the best
pool.}
\end{enumerate}
The derivations, and the certificates that prove the four contracts
optimal on the stated domains, are in Supplemental
Proposition~\ref{prop:MES_certificates}.

In regime~II both levels are proportional to the efficient scale, with
$\qH^*/\qL^*=1+\sqrt2$ for every $a$. The two levels are on opposite
sides of the threshold and their average is the threshold; this is the
incentive constraint. The interval between them, of width
$2(\sqrt2-1)a$, is offered to no type: at the efficient scale all firms
value marginal output equally, so quantities near it do not separate
types. On the upper flat every type receives more than its relaxed
quantity.

Regime~I stands apart. With the contract confined to one side of $\qlim$ the
mirror does not exist, a single pair of types binds, and the separating
branch is $\qiso$, parallel to $\qrel$; in both crossing cases the branch is
tied to $\qlim$ instead.

\section{Discussion}\label{sec:conclusion}
\subsection{The forty configurations}
\label{sec:testbed}
\label{sec:cert_reduction_testbed}

\new{The Supplemental Appendix realizes every configuration of
Assumptions~\ref{S1}--\ref{S_transverse} in a family with affine
$v_{q\theta}$ and $\vs_q$. There are forty configurations, and the system
of Section~\ref{sec:general_apply} produces the reported allocation in
each family. The reduction lemma turns the certificate verification into
a few sign conditions, and by
Corollary~\ref{cor:exact_nocross} the certificate is a characterization
in the no-crossing rows and automatic in the horizontal and mirror rows.
Thirty-four configurations have analytical or exact certificates; five
are partially numerical validations---grid implementability in cases 3,
5 and 13, and 4000-point sampling of weight signs in cases 18 and 20
(Supplemental Section~\ref{sec:cases}). Case 17, with an excluded stretch,
remains an open numerical candidate, reported with its value and without certificate
(Remark~\ref{rem:limits}). Two configurations were first solved by
comparing candidates, full pooling in one and two flat pieces in the
other, and the solutions passed every numerical check; the certificate
failed on both, Proposition~\ref{prop:detection} returned improving
directions, and the optimum is in $D_e$, with gains of order $10^{-3}$
and $10^{-4}$, below the resolution of a grid.
Table~\ref{tab:solutions} of Part~II of the Supplementary Material records the shape,
the value and the certificate of every configuration; the figures and
the scripts are in the deposit.}

\subsection{Concluding remarks}

We have studied screening problems in which the Spence--Mirrlees
condition fails along a monotone dividing curve. Each result of
Section~\ref{sec:continuity}, under its own hypotheses, makes the jump
impossible, unavoidable, or a choice. In the certified
configurations and applications studied here, bunching appears whenever
the relaxed solution is not implementable,
from the erosion of informational rents near the dividing curve.
Optimality is proved by a measure on the family of constraints that bind
together, whose weights are given by the variational system. Geometry
alone does not determine which certificate applies, and two
configurations related by a symmetry of the class need not have the
same solution. The method is not restricted to the linear-curves
family: the two applications have nonlinear primitives and no
closed-form solution.

Four questions remain open. The first is the certificate without
Assumption~\ref{S4}. When $\qrel$ is not monotone, types are paired
discretely rather than in intervals, the weight $R\dens$ is the object
that \citet[eq.~(11)]{AM2010} equalize across paired types, and the
question is which conditions on the pairing replace the monotonicity of
the weight. \ser{The second is strong duality. Theorem~\ref{prop_suff} has no
general converse. Section~\ref{sec:exactness} gives one on the
linear-curves family for three of the four structures; the extreme-type
structure remains, and the two examples of Remark~\ref{rem:limits} show
that the gap can be a sign or a curvature. In both examples the
reported contract is the optimum of every discretization, which points
to a relaxation that is not tight rather than to a better contract;
whether another dual object closes them, or the gap is real, is open.} The third is stochastic: lotteries
are dominated wherever the certificates apply, and \new{where
\eqref{eq:S5} fails we do not know whether a lottery can gain.} The
fourth is dynamic: when the dividing curve moves over time, the jump
becomes a switching time and the comparison among candidates becomes an
optimal stopping problem, with the static trichotomy as boundary
condition.

\begingroup\singlespacing
\bibliographystyle{ecta}
\bibliography{references}
\endgroup
\appendix
\section{Proofs}
\label{app:proofs_general}
\label{app:dem_gateaux}

This appendix collects the short proofs in the order of the text:
Lemma~\ref{lmc}, Theorem~\ref{proposicion},
Proposition~\ref{prop_mirror_shape}, Theorem~\ref{prop_suff} and
Corollary~\ref{prop_stochastic}. The proof of
Proposition~\ref{prop_jump_forced}, the truncation lemma, the implementability claims of
Section~\ref{sec:shape_mirror}, the constructions behind the four
certificate corollaries, and the systems of Proposition~\ref{prop_char}
are in the Supplemental Appendix, Sections~\ref{sec:proofs_main},
\ref{sec:proofs_char} and~\ref{sec:corners}.

\subsection{Shape}
\label{app:shape}

\begin{sergei}
\begin{proof}[Proof of Lemma~\ref{lmc}]
By right continuity there is $\varepsilon>0$ such that $v_{q\theta}>0$
on the segments $[q(x),q(y)]\times\{z\}$ for all $x,y,z\in[\theta,\theta+\varepsilon)$.
For $x<y$ in that interval, adding $\Phi(y,x;q)\ge0$ and
$\Phi(x,y;q)\ge0$ gives $\int_x^y\int_{q(x)}^{q(y)}v_{q\theta}(\xi,z)\,d\xi\,dz\ge0$,
hence $q(y)\ge q(x)$. The other sign is symmetric.
\end{proof}
\end{sergei}

\begin{proof}[Proof of Theorem~\ref{proposicion}]
Let $q\in\mathcal C$ maximize \eqref{maxi} over $\mathcal C$. By
Supplemental Lemma~\ref{lem:truncation}, $q\ge\qrel$ almost everywhere;
since $q$ is c\`adl\`ag and $\qrel$ continuous, $q(\theta)\ge\qrel(\theta)$
and $q(\theta^-)\ge\qrel(\theta)$ at every type. Two facts are used.

\emph{(a) Pairs in reverse order.} If $r$ is non-decreasing with graph in
$\csplus$ and $e<a$, then
$\gif{a}{e}{r}=\int_e^a[v_\theta(r(a),x)-v_\theta(r(x),x)]\,dx\ge0$,
because $\qlim(x)<r(x)\le r(a)$ and $v_\theta(\cdot,x)$ is increasing
above $\qlim(x)$.

\emph{(b) The ratio window.} For $\qrel(\tau)\le l<z<u$ let
\[
A=\int_l^z\!-\vs_q(\xi,\tau)\,d\xi,\ \,
B=\int_z^u\!-\vs_q(\xi,\tau)\,d\xi,\ \,
C=\int_l^z\! v_{q\theta}(\xi,\tau)\,d\xi,\ \,
D=\int_z^u\! v_{q\theta}(\xi,\tau)\,d\xi .
\]
All four are positive, since $\vs_q<0$ above $\qrel$ and $v_{q\theta}>0$
above $\qlim$. By Assumption~H the ratio $R=-\vs_q/v_{q\theta}$ is strictly
increasing on $[l,u]$, and $A/C$ and $B/D$ are averages of $R$ over
$[l,z]$ and $[z,u]$ with weights $v_{q\theta}$; hence $A/C<B/D$, that is,
$A/B<C/D$.

Suppose $q$ jumps at $\tau\in(\uth,\oth)$, with $l=q(\tau^-)<u=q(\tau)$.
Fix $z\in(l,u)$ and $L\in(A/B,C/D)$, and for $\varepsilon>0$ let
$q_\varepsilon=z$ on $I_\varepsilon=[\tau-\varepsilon,\tau+L\varepsilon)$
and $q_\varepsilon=q$ elsewhere. For small $\varepsilon$,
$q_\varepsilon\in\mathcal C$: it is non-decreasing because $q\le l$ below
$\tau$ and $q\ge u$ above, and $z>\qlim$ on $I_\varepsilon$ because
$z>l\ge\qrel(\tau)>\qlim(\tau)$. By continuity of $\vs$ and $\dens$ and
the one-sided limits of $q$ at $\tau$,
\[
\jfunc(q_\varepsilon)-\jfunc(q)=\dens(\tau)\,(BL-A)\,\varepsilon+o(\varepsilon)>0 .
\]

It remains to show that $q_\varepsilon$ is implementable for small
$\varepsilon$. Pairs in reverse order are covered by (a). For pairs
$a<e$, let
\[
M_\varepsilon=\int_{\tau-\varepsilon}^{\tau}[v_\theta(z,x)-v_\theta(q(x),x)]\,dx,
\qquad
N_\varepsilon=\int_{\tau}^{\tau+L\varepsilon}[v_\theta(q(x),x)-v_\theta(z,x)]\,dx,
\]
the rent created on the left of $\tau$ and destroyed on the right. Both
are non-negative, $M_\varepsilon/\varepsilon\to C$ and
$N_\varepsilon/\varepsilon\to LD$, so $M_\varepsilon-N_\varepsilon>0$ for
small $\varepsilon$. If $a\ge\tau+L\varepsilon$, nothing changes. If
$a<\tau-\varepsilon$, the slack of $(a,e)$ changes by
$\int_{I_\varepsilon\cap[a,e]}[v_\theta(z,x)-v_\theta(q(x),x)]\,dx$, which is
non-negative when $e\le\tau$ and at least $M_\varepsilon-N_\varepsilon>0$
when $e>\tau$. If $a\in I_\varepsilon$, the slack is zero when
$e\in I_\varepsilon$, and otherwise
\[
\gif{a}{e}{q_\varepsilon}=\int_{\tau+L\varepsilon}^{e}[v_\theta(q(x),x)-v_\theta(z,x)]\,dx
=G_e(z)-N_\varepsilon,
\]
where $G_t(y):=\int_\tau^t[v_\theta(q(x),x)-v_\theta(y,x)]\,dx$.
Now $G_t(u)=\gif{\tau}{t}{q}\ge0$, $G_t(l)=\lim_{a\uparrow\tau}\gif{a}{t}{q}\ge0$,
and $G_t''(y)=-\int_\tau^t v_{qq\theta}(y,x)\,dx<0$, so $G_t(z)>0$ for every
$t>\tau$. Choose $\delta>0$ with $z>\qlim$ on $[\tau,\tau+\delta]$. For
$e\le\tau+\delta$ the integrand above is non-negative, since
$q(x)\ge u>z>\qlim(x)$. For $e\ge\tau+\delta$, $G_e(z)\ge
k_\delta:=\min_{t\in[\tau+\delta,\oth]}G_t(z)>0$ by continuity in $t$, and
$N_\varepsilon<k_\delta$ for small $\varepsilon$. Hence $q_\varepsilon$ is
implementable, belongs to $\mathcal C$ and earns strictly more than $q$, a
contradiction. So $q$ is continuous on $(\uth,\oth)$, at $\uth$ by right
continuity, and at $\oth$ under the stated normalization.
\end{proof}

\begin{proof}[Proof of Proposition~\ref{prop_mirror_shape}]
\emph{(i)} On $[\uth,\cutoff)$ the allocation lies in $\csminus$, so by
Lemma~\ref{lmc} it is non-increasing on a right neighbourhood of every
type there. Being also non-decreasing, it is locally constant, hence
constant on the interval, and the c\`adl\`ag representative fixes the
value.

\emph{(ii)} By \eqref{eq:G_main}, if $q(\cutoff)\in(\flatq,\qmir(\cutoff))$
then $h(q(\cutoff),\cutoff)<0$, and by right continuity $h(q(x),x)<0$ on
$[\cutoff,\cutoff+\varepsilon)$, so $G_q(\cutoff+\varepsilon)<0$,
contradicting implementability.

\emph{(iii)} \emph{Implementability of $r$.} On $[\cutoff,\oth]$, $r$ is
non-decreasing (the maximum of two non-decreasing functions) and lies
above $\qlim$. Four groups of pairs. Anchor and type in the pool: slack
zero. Anchor in the pool, type $e\ge\cutoff$: by \eqref{eq:G_main},
$G_r(e)=0$ for $e\le\theta_3$ and $G_r$ is non-decreasing afterwards,
since $h(\qrel(x),x)\ge0$ where $\qrel\ge\qmir$; so $G_r\ge0$. Both in
the branch, $s<t$: the sublevel set of $v_\theta(\cdot,t)$ at the level
$v_\theta(r(t),t)$ is an interval $[\ell(t),r(t)]$ with $\ell(t)=\flatq$
on the mirror and $\ell(t)\le\flatq$ on the tail; as
$\flatq<\qmir(\cutoff)\le r(s)\le r(t)$, we get
$v_\theta(r(s),t)\le v_\theta(r(t),t)$ and the forward integrand is
non-negative. Both in the branch, $e<a$: the oriented integral is
$\int_e^a[v_\theta(r(a),x)-v_\theta(r(x),x)]\,dx\ge0$, because
$\qlim(x)\le r(x)\le r(a)$ and $v_\theta(\cdot,x)$ increases above
$\qlim(x)$. Anchor $a$ in the branch, type $e$ in the pool: split at
$\cutoff$; the piece above $\cutoff$ is the previous case, and on the
piece below we need $v_\theta(r(a),x)\ge v_\theta(\flatq,x)$ for
$x<\cutoff$. If $\flatq\ge\qlim(x)$ this is monotonicity of $v_\theta(\cdot,x)$
above $\qlim(x)$, since $r(a)>\flatq$; if $\flatq<\qlim(x)$ it is
$r(a)\ge\qmir(x)$, which holds because $r(a)\ge\qmir(\cutoff)\ge\qmir(x)$
by monotonicity of the mirror.

\emph{The dual bound.} Let $\mu$ be the measure on $(\cutoff,\theta_3]$
\[
d\mu(e)=-\Lambda'(e)\,de+\Lambda(\theta_3)\,\delta_{\theta_3}(de),
\]
finite and non-negative by (M) and $\Lambda\ge0$, with tail mass
$\mu((x,\theta_3])=\Lambda(x)$ for $x\in[\cutoff,\theta_3)$. Take any
implementable $q$ with $q=\flatq$ on the pool. By \eqref{eq:G_main},
$G_q(e)=\gif{a}{e}{q}\ge0$ for every $e$, so
\[
\jfunc(q)\ \le\ \jfunc(q)+\int G_q(e)\,d\mu(e)\ =:\ \mathcal L(q).
\]
Exchanging the order of integration (the integrand is bounded and
$\mu$ is finite) and extending $\Lambda:=0$ on $(\theta_3,\oth]$,
\begin{multline*}
\mathcal L(q)=\int_{\uth}^{\cutoff}\vs(\flatq,x)\dens(x)\,dx
+\int_{\cutoff}^{\oth}\bigl[\vs(q(x),x)\dens(x)+\Lambda(x)\,v_\theta(q(x),x)\bigr]dx\\
-\int_{\cutoff}^{\theta_3}\Lambda(x)\,v_\theta(\flatq,x)\,dx .
\end{multline*}
The last term does not depend on $q$. For each $x$, the map
$z\mapsto\vs(z,x)\dens(x)+\Lambda(x)v_\theta(z,x)$ is strictly concave
on $Q$ by (C) (by Assumption~\ref{S2} where $\Lambda=0$), and its
stationary point solves $\vs_q\dens+\Lambda v_{q\theta}=0$: by
\eqref{eq:Lambda_mirror} this is $z=\qmir(x)$ on $[\cutoff,\theta_3)$,
and $z=\qrel(x)$ where $\Lambda=0$, which is where $\qrel\ge\qmir$ by the
single-meeting hypothesis. Hence $r(x)$ is its unique maximizer
for almost every $x$, and $\mathcal L(q)\le\mathcal L(r)$, strictly
unless $q=r$ a.e.

\emph{Complementary slackness.} $G_r\equiv0$ on $[\cutoff,\theta_3]$,
which carries $\mu$, so $\mathcal L(r)=\jfunc(r)$. Chaining,
$\jfunc(q)\le\mathcal L(q)\le\mathcal L(r)=\jfunc(r)$.
\end{proof}

\subsection{Sufficiency}
\label{app:sufficiency}

\label{app:proofs_suff}

\begin{proof}[Proof of Theorem~\ref{prop_suff}]
	Put $\widehat W(q):=\int_{\Theta^2}\gif{\theta_a}{\theta_e}{q}\,d\mu$ and
	$\mathcal{L}:=\jfunc+\widehat W$. Condition (a) gives complementary
	slackness, $\widehat W(\qopt)=0$; $\gif{\cdot}{\cdot}{q}\ge0$ for
	implementable $q$ and $\mu\ge0$ give weak duality,
	$\jfunc(q)\le\mathcal{L}(q)$. Writing each constraint as the integral over $[\theta_a,\theta_e]$ of
	$v_\theta(q(x),x)-v_\theta(q(\theta_a),x)$, with bounded integrand and
	finite $\mu$, Fubini gives
	$\mathcal{L}(q)=\int_\Theta\mathcal{D}(q(\theta),\theta)\,d\theta
	+\sum_a\mathcal S_a(q(a))$: the integral collects the straddling mass
	\eqref{eq:straddle} and continuous anchor term, while the sum runs over
	atomic anchors. Condition (b) maximizes the latter and (c)--(d) give
	$\mathcal{L}(q)\le\mathcal{L}(\qopt)$, strictly unless $q=\qopt$ a.e.
	Chaining, $\jfunc(q)\le\mathcal{L}(q)\le\mathcal{L}(\qopt)=\jfunc(\qopt)$.
	The four corollaries of Section~\ref{sec:exante} are this argument with
	the measure of each family; each proof verifies (a)--(d) (Supplemental
	Sections~\ref{sec:proofs_main} and~\ref{sec:corners}).
\end{proof}

\begin{proof}[Proof of Corollary~\ref{prop_stochastic}]
	Transfers enter linearly, so for a profile of lotteries $(\sigma_\theta)$
	the envelope theorem of \citet{MilgromSegal2002} gives
	$U'(\theta)=\int v_\theta(\xi,\theta)\,d\sigma_\theta(\xi)$ a.e.;
	expected profit is the expected virtual surplus, and each constraint
	$\Phi$, with both integrals taken against the lotteries, remains
	non-negative. The chain of Theorem~\ref{prop_suff} applies word for
	word: $\jfunc(\sigma)\le\int_\Theta\!\int\mathcal{D}(\xi,\theta)\,
	d\sigma_\theta(\xi)\,d\theta\le\int_\Theta\mathcal{D}(\qopt(\theta),\theta)\,d\theta
	=\jfunc(\qopt)$, the middle step by strict concavity of
	$\mathcal{D}(\cdot,\theta)$, with equality only if $\sigma_\theta$ is
	degenerate at $\qopt(\theta)$ for a.e.\ $\theta$; the atomic terms
	satisfy $\int\mathcal S_a\,d\sigma_a\le\mathcal S_a(\qopt(a))$ by (b).
\end{proof}

\end{document}